\documentclass[sigconf, nonacm]{acmart}

\newcommand\vldbdoi{XX.XX/XXX.XX}
\newcommand\vldbpages{XXX-XXX}
\newcommand\vldbvolume{14}
\newcommand\vldbissue{1}
\newcommand\vldbyear{2020}
\newcommand\vldbauthors{\authors}
\newcommand\vldbtitle{\shorttitle} 
\usepackage{algorithm}
\usepackage{algpseudocode}
\usepackage{subfigure}
\usepackage{placeins}
\usepackage{subcaption}
\usepackage{tabularx} 
\usepackage{caption}  
\usepackage{graphicx}

\newcommand\vldbavailabilityurl{https://github.com/sohrabnamazinia/LangChain}
\newcommand\vldbpagestyle{plain} 

\begin{document}
\title{Personalized Top-k Set Queries Over Predicted Scores}
%with User Defined Scoring Functions}

%%
%% The "author" command and its associated commands are used to define the authors and their affiliations.
\author{Sohrab Namazi Nia, Subhodeep Ghosh, Senjuti Basu Roy}
\affiliation{%
  \institution{NJIT,Newark, NJ, USA}
}
\email{{sn773,sg2646, senjutib} @njit.edu}

\author{Sihem Amer-Yahia}
\affiliation{%
  \institution{CNRS, Univ. Grenoble Alpes, France}
}
\email{sihem.amer-yahia@univ-grenoble-alpes.fr}


\begin{abstract}
This work studies the applicability of expensive external oracles such as large language models in answering top-$k$ queries over predicted scores. Such scores are incurred by user-defined functions to answer personalized queries over multi-modal data. We propose a generic computational framework that handles arbitrary set-based scoring functions, as long as the functions could be decomposed into constructs, each of which sent to an oracle (in our case an LLM) to predict partial scores. At a given point in time, the framework assumes a set of responses and their partial predicted scores, and it maintains a collection of possible sets that are likely to be the true top-$k$. Since calling oracles is costly, our framework judiciously identifies the next construct, i.e., the next best question to ask the oracle so as to maximize the likelihood of identifying the true top-$k$. We present a principled probabilistic model that quantifies that likelihood. We study efficiency opportunities in designing algorithms. We run an evaluation with three large scale datasets, scoring functions, and baselines.  Experiments indicate the efficacy of our framework, as it achieves an order of magnitude improvement over baselines in requiring LLM calls while ensuring result accuracy. Scalability experiments further indicate that our framework could be used in large-scale applications.
\end{abstract}
\maketitle

\pagestyle{\vldbpagestyle}
\begingroup\small\noindent\raggedright\textbf{PVLDB Reference Format:}\\
\vldbauthors. \vldbtitle. PVLDB, \vldbvolume(\vldbissue): \vldbpages, \vldbyear.\\
\href{https://doi.org/\vldbdoi}{doi:\vldbdoi}
\endgroup
\begingroup
\renewcommand\thefootnote{}\footnote{\noindent
This work is licensed under the Creative Commons BY-NC-ND 4.0 International License. Visit \url{https://creativecommons.org/licenses/by-nc-nd/4.0/} to view a copy of this license. For any use beyond those covered by this license, obtain permission by emailing \href{mailto:info@vldb.org}{info@vldb.org}. Copyright is held by the owner/author(s). Publication rights licensed to the VLDB Endowment. \\
\raggedright Proceedings of the VLDB Endowment, Vol. \vldbvolume, No. \vldbissue\ %
ISSN 2150-8097. \\
\href{https://doi.org/\vldbdoi}{doi:\vldbdoi} \\
}\addtocounter{footnote}{-1}\endgroup

%%% VLDB b\inlock start %%%
\ifdefempty{\vldbavailabilityurl}{}{
\vspace{.3cm}
\begingroup\small\noindent\raggedright\textbf{PVLDB Artifact Availability:}\\
The source code, data, and/or other artifacts have been made available at \url{\vldbavailabilityurl}.
\endgroup
}

\section{Introduction}
In several emerging applications that lie at the intersection of databases and machine learning, there is a need for efficient top-$k$ queries, where the scoring function is user-defined. For example, consider a web search query asking for "Hotels with Unique Decor in Manhattan." The task is to select the top-$k$ hotels from a database of user-reviewed entities, based on a scoring criterion related to "unique decor." Since the database may not directly include attributes such as "unique decor," determining such a score becomes challenging. By analyzing user data, this information can be inferred. External oracles/experts, like Large Language Models (LLMs), can provide these predictions, but querying LLMs is costly. \textit{In this paper, we propose a framework that intelligently selects the next best question to ask the oracle, such that the answer maximizes the likelihood of accurately identifying the true top-$k$ set. }

%Similarly, in image retrieval~\cite{DBLP:journals/pvldb/KangEABZ17,DBLP:journals/pvldb/KangGBHZ20}, or in the medical domain, a  query may look for $k$ patients  whose predicted clinical condition is similar to an input patient using a DNN~\cite{DBLP:journals/pvldb/DingAL22,DBLP:conf/cikm/RodriguesPGA20,DBLP:journals/isci/RodriguesGSBA21}. These queries are challenging because finding high quality answers by invoking an expensive oracle such as a human expert, a Deep Neural Network (DNN) model, or a Large Language Model (LLM), on every single entity in the DB and then applying the query, can be prohibitive. 

The main focus of query processing over ML models has been to ensure efficiency without compromising accuracy~\cite{9094012,DBLP:journals/pvldb/DingAL22,DBLP:journals/pvldb/KangEABZ17}. 
Most existing work relies on sampling to reduce the cost of oracle calls in query processing~\cite{DBLP:journals/pvldb/DingAL22,DBLP:journals/pvldb/KangGBHZ20,DBLP:conf/sigmod/GaoXAY21,10.1145/3448016.3452786}.
For instance, the work in~\cite{DBLP:journals/pvldb/DingAL22} proposes to solve the problem of minimizing  oracle usage for finding answers that meet a precision or recall target with provable statistical guarantees while achieving a precision or recall target. {\em To the best of our knowledge, no existing work has tackled the problem of answering personalized top-$k$ set based queries with arbitrary user-defined scoring functions over multi-modal data. Specifically, there is no work that explores how to minimize the number of LLM calls required to score potential top-$k$ sets, where the scoring function is decomposed into a set of constructs, and the score for each construct can be predicted by LLMs.} 

\smallskip \noindent {\bf Running Example.}
We assume a hotel database consisting of $5$ unique entities, i.e., hotels as shown in \autoref{tab:ny_hotels_relevance}. Each entity is represented using multiple items of different modalities (image/audio/text based reviews). Given a query (e.g., unique decor hotels in Manhattan), the goal is to identify a small set of $3$ ($k=3$) hotels that are most suitable.

The process leverages a user-defined set-based scoring function to identify the top-$3$ hotels. We consider scoring functions that are decomposable into constructs, such as, relevance, diversity, serendipity~\cite{div1,div2,div3,div4}, etc. Wlog, an example scoring function with $\ell=2$ constructs, would be used to compute the top-$k$ set that maximizes the sum of relevance and diversity $\mathcal{F}(s, q)= \Sigma_{e \in s} Rel(q,e)+ \Sigma_{e_i,e_j \in s} Div(e_i,e_j)$. The constructs of the scoring function in this example are unary (e.g., relevance) and binary (e.g., diversity). However, we must support constructs of arbitrary arity. 

Since the data is multimodal and the query is personalized, our framework aims to study the applicability of LLMs to retrieve scores of some entities for some scoring construct. %Imagine one such scoring function is the summation of relevance and pairwise diversity of entities within a possible result set. Given the query \( q \) = "Unique decor hotels in Manhattan", suppose there are five hotels , and we wish to find the best set of $3$ hotels. 
% LLMs or generally oracles  could be called to predict the relevance score between a query and an entity, or the diversity score between two entities. Each of these is construed as a question. 
%For example, the oracle could be asked to predict the relevance score of the query ``affordable hotels in Manhattan'' for an entity e.g. ``Hilton NY Hotel'' (a hotel in Manhattan).

%Imagine one such scoring function is the summation of relevance and pairwise diversity of entities within a possible result set. Given a query, such a scoring function identifies top-$k$ set by returning entities that are highly relevant to the query, yet highly diverse from each other. 

%Consider a snapshot of a framework of this sort - at a given point in time, a partial set of scores that are part of the scoring function are known, and the remaining are unknown. Since calling an oracle is costly, the overarching objective of this work is to design a computational framework that identifies the top-$k$ set of entities correctly, while minimizing the number of questions asked to oracles. The goal is to generalize the solution framework that can handle multiple scoring functions, as well as handle uncertainty involved in the response of oracles.

%In \autoref{fig:motivating_scenario}, we illustrate this flow between the data, the computational framework, and the oracle to find the query answer. 

\begin{figure}[ht]
    \centering
    \includegraphics[width=0.5\textwidth]{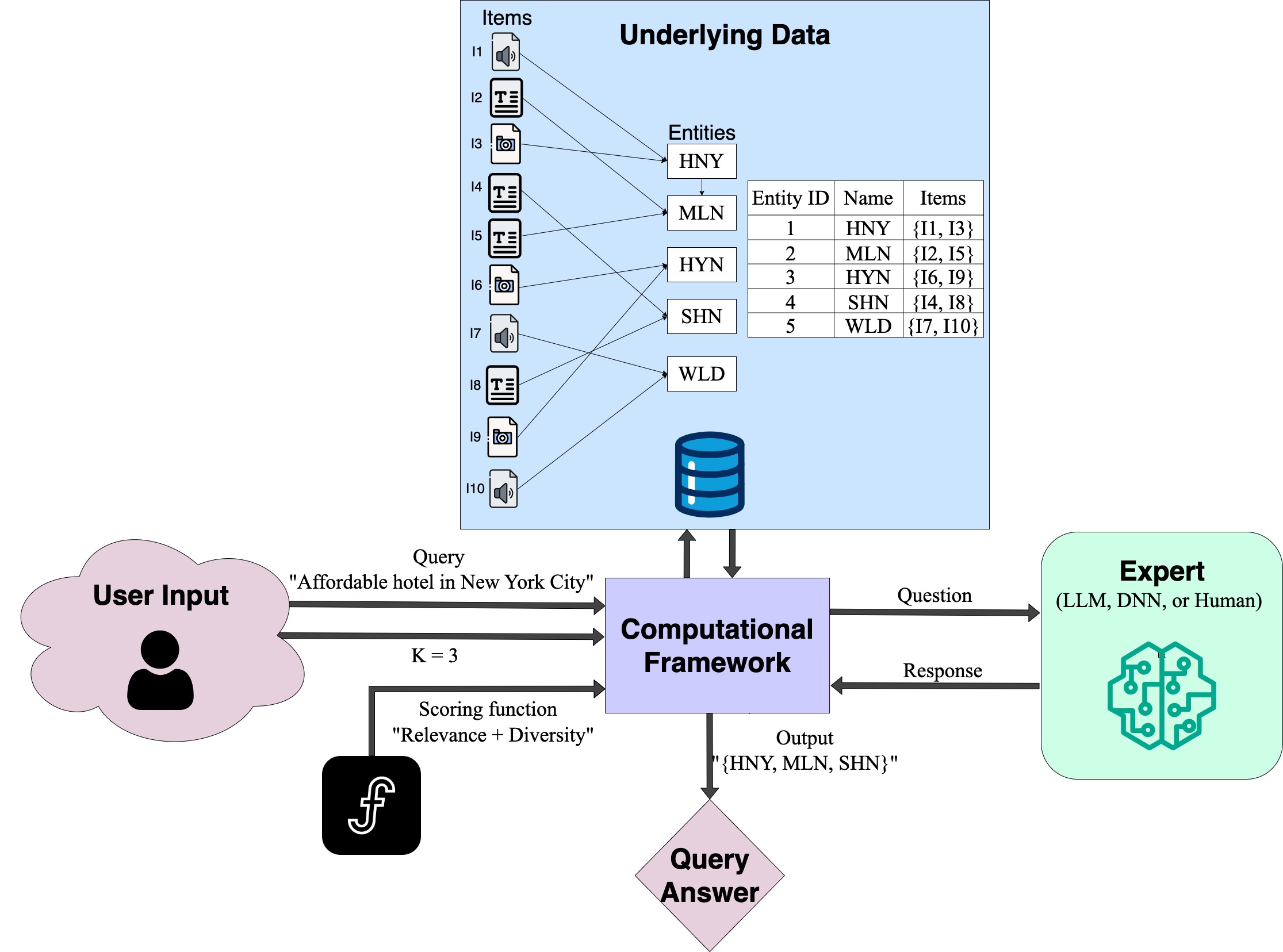}
    \caption{\small Proposed Framework}
    \label{fig:motivating_scenario}
\end{figure}

\smallskip \noindent {\bf Challenges.} 
The overarching challenge is to support any set-based queries with user-defined scoring functions while minimizing the number of calls made to LLMs. The goal is to return the exact answer. Imagine a snapshot of the process - where at a given point in time, partial scores of the set of possible top-k sets (each such set is called a candidate) are computed. Based on those partial scores, the key challenge is to develop a principled ``model'' that quantifies the likelihood of each candidate being the actual top-$k$ answer. The actual query answer therefore is a probability distribution function (pdf), where each candidate is a possible outcome and has a probability associated. Designing this model is non-trivial as computing the probability of a candidate  being the actual top-$k$ is not {\em independent} from other sets that share common entities with it. Given the probabilistic model, the next question to the LLM should be the one that minimizes {\em uncertainty} in the current pdf. Therefore, a challenge is to be able to model ``uncertainty'' in the current pdf and identify the next question that minimizes it. %The final challenge is to study the computational efficiency of the proposed models and designed algorithms.

\smallskip \noindent {\bf Contributions.} We propose a computational framework that works for arbitrary user-defined scoring functions, as long as the scoring function can be decomposed into a set of constructs for which LLM calls could be invoked. At a given point in time, the framework has access to a set of $M$ candidates (each containing $k$ entities), one of which will eventually become the final query answer. The framework can be summarized into $4$ well-defined tasks:\\
\indent {\bf Task 1: Computing score bounds of candidate sets.} At a given time, the framework knows either the partial or the full score of a candidate. When only partial score of a candidate is known, it computes the lower and upper bounds of a candidate's score. \\ %These bounds represent the smallest and largest possible scores a candidate could achieve. \\
\indent {\bf Task 2: Probabilistic model for finding the answer.} This task makes a significant contribution by introducing a principled probabilistic model that estimates the likelihood of each candidate being the true top-$k$. When the score of a candidate is only partially known, it is treated as a discrete pdf with values uniformly distributed between the lower and upper bounds. The likelihood that one candidate having a higher score than the other is computed using {\em max-convolution} of pdfs~\cite{rahman2015worker}. This task considers two cases: computing max-convolution assuming {\bf independence} between candidates, and accounting for {\it dependence} between candidates when they share common entities. Ultimately, this task formalizes the probability of each candidate being the winner or the actual answer. The algorithms developed for this task address both memory and computational efficiency, carefully considering potential bottlenecks in their design.\\
\indent  {\bf Task 3: Determining the next question.} The  query answer is modeled as a random variable with $M$ possible outcomes (one per candidate), and Task 2 formalizes how to compute the pdf of each outcome. The most widely used mathematical framework for quantifying uncertainty is entropy~\cite{renyi1961measures}, with Shannon entropy in information theory being the most common, and the one we adopt in this work. Therefore, the next question to ask the oracle is the one that maximally reduces the uncertainty of this random variable (when entropy is $0$, no further question needs to be asked). We make no further assumption about the possible responses to a question asked to the LLM and design the algorithm for this task with an emphasis on computational efficiency.\\
\indent {\bf Task 4: Response processing.} This task involves aggregating responses from the oracle. This depends on whether the oracle provides a discrete answer or a range. 

For each of these tasks, we study computational challenges and design efficient solutions. 

\smallskip \noindent {\bf Empirical Evaluations.}  We present a comprehensive experimental evaluation of our framework using three large  datasets involving images, reviews, and structured data and three user-defined scoring functions. We implemented multiple baseline solutions as appropriate. Our experimental results indicate the efficacy of the proposed probabilistic models, as it achieves an order of magnitude improvement over baselines in requiring LLM calls while ensuring the achieved results are fully accurate. We measure the number of LLM calls and find that our method for determining the next question significantly reduces the cost compared to baselines. We also find that while the cost increases with an increasing $k$ or number of candidates, our scalability experiments indicate that the proposed models and algorithms could be used inside large-scale applications.

The rest of the paper is organized as follows: Section~\ref{sec:dm} presents our data model and problem definitions. Sections~\ref{sec:framework} and~\ref{sec:alg} contains our framework and its algorithms. We provide experimental evaluations in Section~\ref{sec:exp}. Related work is summarized in Section~\ref{sec:rel}. Several extensions of the proposed work are discussed in Section~\ref{sec:ext} and we conclude in Section~\ref{sec:conc}.
\section{Data Model \& Problem Definition}\label{sec:dm}
%We describe our data model and formalize our problem.

\subsection{Data Model}
\noindent \textbf{Input Query:} A user writes a query $q$ with an input parameter $k$, with the goal of obtaining the top-$k$ set of entities to $q$.

In our running example, $q$ = "Hotels with Unique Decor in Manhattan" and we need to find the top-$3$ hotels with respect to $q$.  

\noindent \textbf{Database.} The query comes to a database $D$ containing data of different modalities, such as, pictures, text, audio, etc. Item $i$ represents a physical entity $e$ (such as a hotel). There may be multiple items associated with an entity. Let $E$ represent the set of unique entities, such that, $|E|=n$.

\autoref{fig:motivating_scenario} contains a snapshot of our example with $5$ entities or hotels. Each entity has two associated items. For instance, the  hotel "HNY", has \{ i1, i3 \}, an audio and an image respectively. \\

\noindent \textbf{Scoring function and characteristics.} A user provides a set-based scoring function $\mathcal{F}$ that admits a query  $q$ and an integer $k$, and scores any subset $s \subseteq E$ entities, such that $|s|=k$, i.e., $\mathcal{F}(s,q) \rightarrow R$. 

A scoring function needs to satisfy two conditions: (a) it is set-based, meaning it computes the score of a set of size $k$ and need not provide the individual order of  entities in the set, (b) it is decomposable into $\ell$ constructs each of which could be sent to an oracle (in our case LLM) to obtain its predicted value. 

Using our running example, $\ell=2$ and contains two constructs: a unary construct $Rel(q,e)$, and a binary construct $Rel(e_i,e_j)$.

\noindent \textbf{Top-$k$ set.} Given a query $q$ and an integer $k$, a set of entities $E$, and a set-based scoring function $\mathcal{F}$, $s_q$ is a top-$k$ set, if $s_q \subseteq E$, $|s_q|=k$, and 
$s_q = argmax_{s \subseteq E: |s|=k} \mathcal{F}(s,q)$.

Given our running example, 
 $c_1 = \{ HNY, MLN, HYN \}$ is a candidate set of top-$3$ hotels. Determining the top-$k$ set requires more insights on the unknown values in \autoref{tab:ny_hotels_relevance} and \autoref{tab:ny_hotels_diversity}.

\noindent \textbf{Candidate.} A candidate $c$ is a set of $k$ entities with score $\mathcal{F}(c,q)$. When the score of $c$ is partially computed, it's between a range, (LB\_c, UB\_c). Score of $c$ is then treated as a uniform distribution with $m$ discrete values between $(LB\_c, UB\_c)$.

Using our example, the set of hotels $c_1=\{ HNY, MLN, HYN \}$ is a candidate, and its score is a uniform probability distribution function (PDF) within the range $(0, 2)$ with 5 discrete values as $0, 0.5, ..., 2.0$.  

\noindent \textbf{Candidate set.} At a given point in time, the process keeps track of $C=\{c_1, c_2, \ldots, c_M\}$ candidates, where $M \leq \binom{n}{k}$. 

Given our running example, $M=3$ and there are three possible candidates, \[
    c_1 = \{\text{HNY, MLN, HYN}\} \]
 \[
    c_2 = \{\text{HNY, MLN, WLD}\},c_3 = \{\text{HNY, HYN, SHN}\}
\]

\noindent \textbf{Questions.} Given a user-defined scoring function $\mathcal{F}$ with $\ell$ constructs, a question $\mathcal{Q}$ asks the predicted score over one of the $\ell$ constructs appropriate inputs.

For our example function, $\mathcal{Q}(e_i,q,Rel)$ asks for the relevance score between the query $q$ and an entity $e_i$, whereas, $\mathcal{Q}(e_i,e_j,Div)$ asks for the diversity score between entities $e_i$ and $e_j$.
 
 %   \item $\mathcal{Q}_2$, which is a question of the form $\mathcal{Q}(e_i, e_j, Div)$.
%\end{itemize}
   %where $e_i,e_j$ are the $i$-th, $j$-th entities, $q$ is the query, and $F_{\ell}$ is the $\ell$-th component of the scoring function. Depending on the nature of the scoring function, some constructs involve returning a score between a query and one of the two items, e.g., to compute relevance, which requires questions of the form $Q_1$. Some other construct may admit two entities, to compute diversity, which requires questions of the form $Q_2$.

Considering our running example, one can propose different questions about the hotels mentioned in \autoref{tab:ny_hotels_relevance}. For instance, $\mathcal{Q}(HNY,q,Rel)$ asks for the relevance between  hotel "HNY" and $q$, and $\mathcal{Q}(WLD,SHN,Div)$ refers to a question about the diversity between hotels "WLD" and "SHN".

% NOTE; figure 1 change sec column to hotel names
\noindent \textbf{Response from oracle (LLM).} Given a question \( \mathcal{Q} \), and a pre-specified range of values \([MIN,MAX]\) which the oracle is instructed to respond within, its response \( \mathcal{R} \) is defined as one of:
\begin{itemize}
    \item \( \mathcal{R}_1 = r \in [MIN,MAX] \), where \( r \) is a discrete value within the specified range \([MIN, MAX]\).
    \item \( \mathcal{R}_2 = (l,u) \), where \( MIN \leq l < u \leq MAX \) is a continuous range of values that the oracle can provide as a response.
\end{itemize}

As an example, the last column of \autoref{tab:ny_hotels_relevance} is an oracle response of type $\mathcal{R}_1$, and each cell in \autoref{tab:ny_hotels_diversity} is an oracle response of type $\mathcal{R}_1$. For the question $\mathcal{Q}(HYN,q,Rel)$, the oracle returns a discrete value $Rel(HNY,q) = 0.5$ for which the $[MIN,MAX]$ are assumed to be $0$ and $1$ respectively. At a given point in time, some responses might be unknown and are denoted by $U$ in the tables. For a question $\mathcal{Q}(WLD,SHN,Div)$, we could have $Div(WLD,SHN) = [0.3,0.7]$ as the response which is a continuous range. 

We study the former kind in depth in this work, and explain how the latter could be adapted in Section~\ref{sec:ext}.

%\textcolor{orange}{we need a small table of notations - for m, M, n, e, c, C, $Q_u$..put it next to table 1 without incurring any more space}

\subsection{Formalism}
Given a scoring function $\mathcal{F}$, we define the set of all possible questions, denoted as \( Q_U \), as the union of $\ell$ decomposable constructs instantiated with input data, as appropriate.

In our example,  \( Q_U \) contains relevance on each of the $5$ hotels in $D$, plus a total of $\binom{5}{2}$ pairwise diversity questions. At a given point in time, the response to some of these questions could be known. 

\textbf{Known information.} Let \( Q_K \) be the set of questions whose responses are known. Thus, we have:
\[
Q_K = \{ q \in Q_{ALL} \mid \text{response to } q \text{ is known} \}
\]

\textbf{Unknown information.} Conversely, we can define the unknown information as the set of questions for which the responses have not yet been obtained. This set, denoted as \( Q_{ALL} \), is given by:
\[
Q_U = Q_{ALL} \setminus Q_K
\]

Using our running example, the response to each relevance and diversity question has been mapped to a cell in \autoref{tab:ny_hotels_relevance} and \autoref{tab:ny_hotels_diversity}, respectively. Cells marked as $U$ indicate an unknown value for that specific question. Hence, based on these two tables, we have:

\[
Q_K = \left\{
\begin{array}{lll}
Rel(\text{MLN}, q) & Rel(\text{WLD}, q) & Rel(\text{HNY}, q) \\
Div(\text{HNY}, \text{MLN}) & Div(\text{MLN}, \text{HNY}) & Div(\text{SHN}, \text{WLD}) \\
Div(\text{HYN}, \text{SHN}) & &
\end{array}
\right.
\]

\[
Q_U = \left\{
\begin{array}{lll}
Rel(\text{HYN}, q) & Rel(\text{SHN}, q) & Div(\text{HYN}, \text{HNY}) \\
Div(\text{HYN}, \text{MLN}) & Div(\text{HYN}, \text{WLD}) & Div(\text{HNY}, \text{HYN}) \\
Div(\text{MLN}, \text{HYN}) & Div(\text{SHN}, \text{HYN}) &
\end{array}
\right.
\]

Given the running example with $C=\{c_1,c_2,c_3\}$, we aim to find the next best question, that is the most useful for finding the query answer. 

\begin{table}[h!]
\centering
\begin{minipage}{0.45\textwidth}
\centering
\begin{tabular}{|c|l|}
\hline
\textbf{Notation} & \textbf{Definition} \\ \hline
$m$ & \# discrete score values \\ \hline
$M$ & \# candidates \\ \hline
$n$ & \# entities \\ \hline
$e$ & Entity \\ \hline
$c$ & Candidate \\ \hline
$C$ & Candidates set \\ \hline
$Q_u$ & Set of unknown questions \\ \hline
\end{tabular}
\caption{\small Notations}
\end{minipage}%
\hspace{1cm}  % Adjust space between the tables
\begin{minipage}{0.45\textwidth}
\centering
\begin{tabular}{|c|l|c|}
\hline
\textbf{Abbreviation} & \textbf{Hotel Name} & \textbf{Relevance} \\ \hline
HNY & Hilton NY      & U \\ \hline
MLN & Marriott LN    & 1.0 \\ \hline
HYN & Hyatt NY       & 1.0 \\ \hline
SHN & Sheraton NY    & 0.0 \\ \hline
WLD & Waldorf NY     & 0.5 \\ \hline
\end{tabular}
\caption{\small Relevance scores of the hotels}
\label{tab:ny_hotels_relevance}
\end{minipage}
\end{table}

\begin{table}[h!]
\centering
\begin{tabular}{|c|c|c|c|c|c|}
\hline
       & HNY  & MLN  & HYN  & SHN  & WLD  \\ \hline
HNY    & -   & 1.0  & 0.5    & 0.5  & 0.5  \\ \hline
MLN    & 1.0  & -    & U    & U  & U  \\ \hline
HYN    & 0.5    & U    & -    & U  & 0.5    \\ \hline
SHN    & 0.5  & U  & U  & -    & U  \\ \hline
WLD    & 0.5  & U  & 0.5   & U  & -    \\ \hline
\end{tabular}
\caption{\small Diversity scores between pair of hotels}
\label{tab:ny_hotels_diversity}
\end{table}

\textbf{Problem Definition:} Given \( n \) entities, an integer $k$, an input query \( q \), a user-defined set-based scoring function \( \mathcal{F} \) with $\ell$ constructs, and already known information \( Q_K \), identify the next question to be asked to the oracle (LLM) that has the highest likelihood of identifying the true top-$k$ set.

We aim to find the next best question to ask, i.e., the one that is most useful in finding the the top-$k$ set with the highest score.
    
Potentially, the next question could be any of the cells marked as U (Unknown) in \autoref{tab:ny_hotels_relevance} or \autoref{tab:ny_hotels_diversity}. As an example, consider two possible questions to ask, and compare how useful they can be:
    
    \begin{itemize}
        \item \( Q_1 \): \(Div(MLN, HYN)\) - the diversity score between "MLN" and "HYN"
        \item \( Q_2 \): \(Div(HYN, SHN)\) - the diversity score between "HYN" and "SHN"
    \end{itemize}

Before asking these questions, by using the known values have:
\[
\begin{aligned}
    \mathcal{F}(c_1, q) = 
    & Rel(\text{HNY},q) + Rel(\text{MLN},q) + Rel(\text{HYN},q) + Div(\text{HNY, MLN}) \\
    & + Div(\text{HNY, HYN}) + Div(\text{MLN, HYN}) \\
    = & 3.5 + Rel(\text{HNY},q) + Div(\text{MLN, HYN}) \\
\end{aligned}
\]

\[
\begin{aligned}
    \mathcal{F}(c_2, q) =  3.0 + Rel(\text{HNY},q) + Div(\text{MLN, WLD}) 
\end{aligned}
\]

\[
\begin{aligned}
    \mathcal{F}(c_3, q) =  2.0 + Rel(\text{HNY},q) + Div(\text{HYN, SHN}) 
\end{aligned}
\]

Suppose we choose to ask \( Q_1 \), revealing that \(Div(MLN, HYN) = 1.0 \). Knowing this additional information, we update the overall score of \( c_1 \) as follows:
\[
\begin{aligned}
    \mathcal{F}(c_1, q) = 
     & 4.5 + Rel(\text{HNY},q) \\
\end{aligned}
\]

Since all the known/unknown values are within a pre-specified range, in this example [0,1], regardless of the values of remaining unknown variables in $\mathcal{F}(c_2, q)$, $\mathcal{F}(c_1, q)$, and $\mathcal{F}(c_1, q)$, we will have:
    \[
\begin{aligned}
    \mathcal{F}(c_1, q) \geq \mathcal{F}(c_2, q) \quad \text{AND} \quad
    \mathcal{F}(c_1, q) \geq \mathcal{F}(c_3, q)
\end{aligned}
\]

    Hence, in this case, we can conclude that $c_1$ is the query answer. By asking \(Div(MLN, HYN)\) from the LLM, we could potentially get a score different from $1.0$ as well. However, in this example, if \(Div(MLN, HYN) \geq 0.5 \), we could still come up with the same conclusion as we did. Therefore, asking \(Div(MLN, HYN)\) can likely result in finding the top-$3$ set.  

 On the other hand, suppose we choose to ask \( Q_2 \), resulting in \(Div(MLN, HYN) = 1.0 \). Knowing this additional information, we update the overall scores of \( c_3 \) as follows:  
    \[
    \begin{aligned}
        \mathcal{F}(c_3, q) = 
         & 3.0 + Rel(\text{HNY},q) \\
    \end{aligned}
    \]

In this case, the top-$3$ set is not final yet. In more detail, based on different possible values of $Div(MLN, HYN)$ and $Div(HYN, SHN)$ corresponding to $c_1$ and $c_2$ respectively, any of them could have a higher final score. Hence, asking $Q_2$ will determine the top-$k$ set. 

    Generally, at a given point in time, there might be more than one question that could likely determine the top-$k$ set. Alternatively, there might be no question that could determine the top-$k$ set at a specific time, but choosing one question over another might result in a fewer number of required questions to ask in the future to determine the top-$k$ set. 
        
    Hence, it is crucial that at each step we ask the question that is most likely going to reduce the uncertainty in finding the top-$k$ set and guide us toward identifying the query answer more efficiently.

%In Section ~\ref{sec:section2.3}, we introduce each component in our proposed framework to solve this problem.

\section{Proposed Framework}\label{sec:framework}
Depending on the scoring function $\mathcal{F}$ the framework identifies one, several, or all candidate top-$k$ sets as answers to a query $q$. Wlog, we assume a set $C$ of $m$ such candidate sets and present $4$ essential tasks to solve our problem.
%Using the running example, two possible candidates could be: 
     %   \[
     %   c_1 = \{\text{HNY, MLN, HYN}\}, \quad c_2 = \{\text{MLN, HYN, WLD}\}
     %   \]
     %%   Hence, we can define the candidates' set C as follows:
     %   \[
     %   C = \{c1, c2\}
     %   \]

\noindent {\bf Task 1. Computing Score Bounds of Candidate Sets.}
At a given point in time, only the partial score of a candidate $c \in C$ is known. 

\noindent {\bf Technical Problem: Lower and upper bounds of score of $c$.} 
Given $\mathcal{F}$, and known information $Q_K$, compute the lower (resp. upper) bound of score of $c$ as the smallest (resp. largest) score of $c$.

Considering our example, to find the lowest and highest possible scores of $c_1$, we should replace known values from \autoref{tab:ny_hotels_relevance} and \autoref{tab:ny_hotels_diversity} in $\mathcal{F}_{\text{min}}(c_1, q)$ and $\mathcal{F}_{\text{max}}(c_1, q)$ and substitute unknown values in the formula with minimum and maximum possible response values, which are assumed to be $0$ and $1$ respectively. Therefore, we can compute the minimum and maximum overall scores of $c_1$ as follows:
\[
\mathcal{F}_{\text{min}}(c_1, q) = Rel(\text{HNY},q) + Rel(\text{MLN},q) + 0 + Div(\text{HNY, MLN}) + 0 + 0
\]
\[
\mathcal{F}_{\text{min}}(c_1, q) = 0.5 + 1.0 + 0 + 0.5 + 0 + 0 = 2.0
\]

\[
\mathcal{F}_{\text{max}}(c_1, q) = Rel(\text{HNY},q) + Rel(\text{MLN},q) + 1 + Div(\text{HNY, MLN}) + 1 + 1
\]
\[
\mathcal{F}_{\text{max}}(c_1, q) = 0.5 + 1.0 + 1 + 0.5 + 1 + 1 = 5.0
\]

Hence, the lower bound (LB) and upper bound (UB) of \( c_1 \) are:
\begin{align*}
    (\text{LB, UB})_{c_1} &= (2.0, 5.0)
\end{align*}
    
\noindent {\bf Task 2: Probabilistic Model for Finding the Answer.} 
 Given a set $\mathcal{C}$ of $m$ candidate top-$k$ sets, the probability $P(c)$ represents that candidate $c$ is the answer ($c^*$) of the query.
 
\noindent{\bf Technical Problem:} The probability of a candidate $c$ being the query answer $c^*$ is:
 \begin{align}\label{eq:prob}
    P(c = c^*) &= P\left(\bigcap_{c_j \in \mathcal{C}} \mathcal{F}(c, q) \geq \mathcal{F}(c_j, q)\right)
\end{align}

%There are two possible scenarios.

\noindent{\bf Independence among candidates.}
In the simplest case, each candidate has unique entities with no entity in common across any two candidates. The joint probability could be calculated as:

\begin{align}\label{eq:ind}
P(c = c^*) &= \prod_{i=1}^{M} P\left(\mathcal{F}(c, q) \geq \mathcal{F}(c_i, q)\right)
\end{align}

\noindent{\bf Dependence among candidates.}
When there exist entities that are common across multiple candidates, the probabilistic model capturing a candidate being the winner (or query answer) needs to account for conditional probabilities, as follows:

%\textcolor{red}{put this in align environment and give it an equation no}

\begin{align}\label{eq:joint}
P(c = c^*) = \prod_{i=1}^{M} P\left(\mathcal{F}(c, q) \geq \mathcal{F}(c_i, q) \mid \bigcap_{j=1}^{i-1} \left( \mathcal{F}(c, q) \geq \mathcal{F}(c_j, q) \right) \right)
\end{align}

Equation~\ref{eq:joint} takes the following form once expanded.
\begin{align}\label{eq:jointexpand}
P(c = c^*) &= P\left(\mathcal{F}(c, q) \geq \mathcal{F}(c_1, q)\right) \times \notag \\
& \quad P\left(\mathcal{F}(c, q) \geq \mathcal{F}(c_2, q)  \mid \mathcal{F}(c, q) \geq \mathcal{F}(c_1, q) \right) \times \notag \\
& \quad \ldots \times P\left( \mathcal{F}(c, q) \geq \mathcal{F}(c_M, q) \mid \mathcal{F}(c, q) \geq \mathcal{F}(c_1, q), \right. \notag \\
& \quad \quad \left. \mathcal{F}(c_2, q) \geq \mathcal{F}(c_2, q), \mathcal{F}(c_2, q) \geq \mathcal{F}(c_3, q), \right. \notag \\
& \quad \quad \left. \ldots, \mathcal{F}(c, q) \geq \mathcal{F}(c_{M-1}, q)\right)
\end{align}

\begin{comment}
This formula is basically a multiplication of probabilities of the following type:
\[
P(\mathcal{F}(c, q) \geq \mathcal{F}(c_i, q) \mid \bigcap_{j=1}^{i-1} \left( \mathcal{F}(c, q) \geq \mathcal{F}(c_j, q) \right))
\]

    Hence, We need to compute probabilities of the following form:

\[
P(f(c_1) \geq f(c_2) \mid \text{conditions})
\]

This probability represents the likelihood that candidate \( c_1 \) scores higher than candidate \( c_2 \), given a set of constraints imposed by previous comparisons between candidates.

To compute this probability, we can rewrite it as follows:
\[
P(f(c_1) > f(c_2) \mid \text{constraints}) = \frac{p((f(c_1) \geq f(c_2)) \land \text{constraints})}{p(\text{constraints})} 
\]
\[
= \frac{n((f(c_1) \geq f(c_2)) \land \text{constraints})}{n(\text{constraints})}
\]

Where the denominator represents the size or number of all possible score assignments to candidates that satisfy given constraints. Similarly, the numerator is just the size of a subset of scoring assignments of the denominator that capture the additional constraint: \( f(c_1) \geq f(c_2) \). 
\end{comment}

Using our running example, 
\begin{align*}
P(c_2= c^*) &= P\left(\mathcal{F}(c_2, q) \geq \mathcal{F}(c_1, q)\right) \times \notag \\
& \quad P\left(\mathcal{F}(c_2, q) \geq \mathcal{F}(c_3, q) \mid \left( \mathcal{F}(c_2, q) \geq \mathcal{F}(c_1, q) \right)\right) 
\end{align*}

\begin{comment}
    In section \ref{winning_probability}, we will elaborate how we compute candidates winning probabilities using the above formula. Given the running example, assuming we only have $c_1, c_2$, and $c_3$ as the candidates, after computing winning probability distribution function, we will have:

\[
P(c = c^*) =
\begin{cases} 
0.75 & \text{ } c = c_1, \\
0.24 & \text{ } c = c_2, \\
0.01 & \text{ } c = c_3.
\end{cases}
\]
\end{comment}

%{\color{blue} The notion of constraint is too abstract.}
%\textcolor{red}{Sohrab: just write the value: Using the running example, $P(c_1)=xx$, such that $c_1=c*$.}

\noindent {\bf Task 3: Determining the Next Question.} 
For a candidate $c$,  $P(c)$ represents the probability of $c$ being the answer. Given the set of candidates $C$, the actual answer $c^*$ is thus a random variable with probability distribution $\mathcal{A}$, representing the probability of each candidate being the answer.
 \begin{definition}
     {\bf Uncertainty in $\mathcal{A}$.} The uncertainty in $c^*$ is modeled as the entropy~\cite{renyi1961measures} in $\mathcal{A}$, as follows:
    \begin{align*}
    H(\mathcal{A}) &= - \sum_{c \in \mathcal{C}} P(c) \log(P(c))
\end{align*}
where $P(c)$ is the probability of candidate $c$ to be the query answer. 
 \end{definition}

\noindent{\bf Technical Problem: Selecting the next best question.} Let $H(\mathcal{A})$ be the entropy associated with the query answer $c^*$. When $Q \in Q_U$ is provided by the oracle, let $H(\mathcal{A'})$ be the reduced entropy, select $Q \in Q_U$ such that $H(\mathcal{A}) - H(\mathcal{A'})$ is maximized. In other words, maximizing the difference between $H(\mathcal{A}) - H(\mathcal{A'})$ is equivalent to minimizing  $H(\mathcal{A'})$. Entropy measures the uncertainty associated with the query answer - thus, minimizing this enhances predictability. In fact, when the entropy is $0$, the query answer could be decided with complete certainty.

Using our running example with the three candidates, $C = \{c_1, c_2, c_3\}$, the associated entropy is $0.604$.

\begin{comment}
    \begin{align*}
    H(\mathcal{A}) &= - (p_{c_1} \log(p_{c_1}) + p_{c_2} \log(p_{c_2}) + p_{c_3} \log(p_{c_3}))
    \end{align*}

\[
    H(c^*) \approx \text{0.604}
\]
\end{comment}

In Section~\ref{subsec:nextquestion}, we will propose algorithms that select the next question to ask s.t. entropy is minimized. Minimizing entropy will minimize uncertainty in finding the query answer $c^*$. Given the running example, we shall show that our solution will choose $Q = Div(HNY, MLN)$ as the next question. Given the LLM returns $Div(HNY, MLN) = 1$, $H(c^*) = 0$ and $c_1$ becomes the query answer.

\noindent {\bf Task 4: Response Processing.} The final task is to process the obtained response from the oracle. There are two obvious possibilities: a. the oracle returns a discrete value, b. the oracle returns a range. There are also other types of responses such as aggregating multiple discrete oracle responses, or aggregating multiple oracle responses each providing a range. 

We study the first case closely in this work. In Section~\ref{subsec:resp}, we discuss how a discrete response (e.g., 0.7) from a single oracle is processed. In Section~\ref{sec:ext}, we discuss other alternatives and how they could be adapted.

%\textcolor{red}{comment from Sohrab: I think the following technical problem definition is for task 2. Task 4 doesn't update bounds. just processes responses from expert(s). since we have mentioned here we study first case closely and discuss others in section 7, maybe no running example is needed here? because first case is simply single expert single response so no further processing is needed. The example i have mentioned for this task in next section is for multiple experts discrete values.  }

%\noindent{\bf Technical Problem: Response processing.} 
%Given $\mathcal{F}$ and the newly acquired response $Q_r$, update score bounds of $C$.

%\textcolor{red}{Sohrab: just write the value: Using the running example, write the task if a specific answer is obtained.}

\section{Algorithms}\label{sec:alg}
We are now ready to provide  algorithms that solve the tasks described in Section~\ref{sec:framework}.

\subsection{Computing Score Bounds of Candidate Sets}
\label{bound_computation}

Given the set of candidates,  we discuss algorithms to compute score bounds of each of them. Since only partial scores of the candidates are known, this process computes the score bounds (LB\_c, UB\_c) of each candidate $c$. These bounds are calculated based on $\mathcal{F}(c, q)$ and the known information $Q_k$ and are updated when $Q_k$ is updated. Given $c$, the process looks at the constructs in $\mathcal{F}(c, q)$, and uses the actual scores for the parts that could be obtained from $Q_k$. For the rest, it uses the minimum possible score to produce LB\_c and the maximum possible score to produce UB\_c. Without any further assumption in place, this gives tight bounds.

To simplify exposition, let us assume that the minimum and the maximum relevance and diversity scores of each construct is within \((0, 1)\). On the other hand, as shown above, there exists a total of $6$ constructs contributing to the overall score of a candidate. Therefore, initially, the LB and UB of any candidate is $0$ and $6$, respectively.

\begin{comment}
    \begin{align*}
\text{LB}(c_i) &= \text{MIN\_SCORE} = 0 \quad \forall c_i \in \text{candidates} \\
\text{UB}(c_i) &= \text{MAX\_SCORE} = 6 \quad \forall c_i \in \text{candidates}
\end{align*}

Hence, We can represent the initial candidates set with their corresponding bounds as follows:

\[
    \text{C} = 
    \begin{cases}
    c_1 = \{ \text{HNY}, \text{MLN}, \text{HYN} \}, & c_2 = \{ \text{HNY}, \text{MLN}, \text{SHN} \} \\
    c_3 = \{ \text{HNY}, \text{MLN}, \text{WLD} \}, & c_4 = \{ \text{HNY}, \text{HYN}, \text{SHN} \} \\
    c_5 = \{ \text{HNY}, \text{HYN}, \text{WLD} \}, & c_6 = \{ \text{HNY}, \text{SHN}, \text{WLD} \} \\
    c_7 = \{ \text{MLN}, \text{HYN}, \text{SHN} \}, & c_8 = \{ \text{MLN}, \text{HYN}, \text{WLD} \} \\
    c_9 = \{ \text{MLN}, \text{SHN}, \text{WLD} \}, & c_{10} = \{ \text{HYN}, \text{SHN}, \text{WLD} \} 
    \end{cases}
\]

\[
\text{C\_bounds} = \{ c_1: (0, 6), c_2: (0, 6), c_3: (0, 6), \ldots, c_{10}: (0, 6) \}
\]

\end{comment}

Now, let us assume \( Rel(WLD, q) \) is obtained and \( Rel(WLD, q) = 0.5 \). Hence, we need to update LB and UB of those candidates containing \( WLD \). As an example, $c_2$ now becomes
\begin{align*}
\mathcal{F}(c_2, q) & = Rel(HNY, q) + Rel(MLN, q) + Rel(WLD, q) + Div(HNY, MLN) \\
          & \quad + Div(HNY, WLD) + Div(MLN, WLD) \\
          & = Rel(HNY, q) + Rel(MLN, q) + 0.5 + Div(HNY, MLN) \\
          & \quad + Div(HNY, WLD) + Div(MLN, WLD) \\
          & \implies 0.5 \leq \mathcal{F}(c_2, q) \leq 5.5
\end{align*}

\subsection{Probabilistic Model for Finding the Answer}
\label{winning_probability}
The algorithm designed for this task, namely finding $c^*$, needs to calculate the probability that the score of $c^*$ is larger than the scores of all other candidates, as defined in Equation~\ref{eq:prob}. We present algorithms for two variants - {\tt ProbInd} assumes independence among candidates and {\tt ProbDep} accounts for potential dependencies among candidates. Since the score of each candidate is a uniform probability distributions within (LB,UB), our solution requires to compute max convolution of probability distributions~\cite{rahman2015worker}, as defined below.

\begin{definition}
{\bf (Max-Convolution of Distributions).}
Assume that $f(c)$ and $g(c_1)$ are the pdfs of the two independent random variables $c$ and $c_1$ respectively. The pdf of the random variable $Max(c, c_1)$ is the max convolution of the two pdfs and is calculated as follows:
$P\left(\mathcal{F}(c_1, q) \geq \mathcal{F}(c_2, q)\right)  = \Sigma_{\forall x, x_1} P(c=x) \times P(c_1=x_1),{ x \in [LB\_c, UB\_c], x_1 \in [LB\_{c_1}, UB\_{c_1}] : x \geq x_1,} $
\end{definition}
Figure \ref{fig:max_conv} shows $Max(c_1, c_2)$ 

\begin{figure}[ht]
    \centering
    \includegraphics[width=0.4\textwidth]{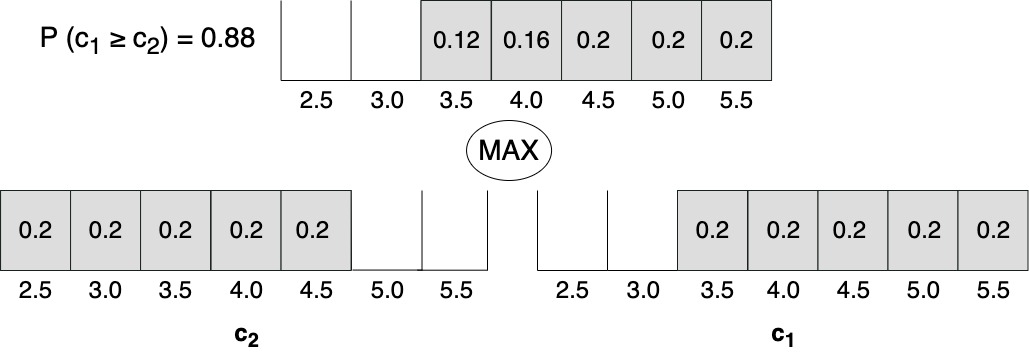}
    \caption{Computing \(\mathbf{P(c_1 \geq c_2)}\)}
    \label{fig:max_conv}
\end{figure}

%In \autoref{uncertainty}, we use this probability distribution function (PDF) to demonstrate what is the best next question to ask expert which maximally reduces the uncertainty of finding the winner candidate. Ultimately, we look for a time step in which $c^*$ is known, which means zero uncertainty in finding the winner candidate. 

%\textbf{Question:} Given a candidates set $C$ with corresponding bounds for each candidate, and a specific candidate $c \in C$, what is the probability that $c = c^*$?

\begin{comment}
In order to find this probability, we consider two different cases. 

\textbf{Case 1 - Independence assumption}: In this case, we assume that the candidates are independent. Hence, we can compute $P(c = c^*)$ as follows:

% example

\[
P(c = c^*) = \prod_{c' \in C} P(\mathcal{F}(c, q) \geq \mathcal{F}(c', q))
\]

\textbf{\textcolor{red}{An example needs to be added here}}

\end{comment}

\noindent \textbf{Case A - Independence among candidates}
Given the set $C$ of $M$ candidates, if every candidate contains entities that are only present in that candidate, the probability of a candidate $c$ being the winner, i.e., $P(c = c^*)$ is the joint probability of $c$ being larger than every other candidate, and could be calculated using Equation~\ref{eq:ind}. Algorithm {\tt ProbInd} does that.

Consider an imaginary example of the following kind $C=\{c_1,c_2,c_3\}$, such that $c_1=\{e_1,e_2,e_3\}, c_2=\{e_4,e_5,e_6\},c_3=\{e_7,e_8,e_9\}$. The probability of $c_1$ being the answer could be expressed as follows:

\begin{align*}
P(c_1 = c^*) &=  P\left(\mathcal{F}(c_1, q) \geq \mathcal{F}(c_2, q)\right) \times P\left(\mathcal{F}(c_1, q) \geq \mathcal{F}(c_3, q)\right)
\end{align*}

Specifically, the score of each candidate $c$ is within the range (LB\_c, UB\_c). Algorithm {\tt ProbInd} treats the score of each candidate $c$ as a uniform distribution within (LB\_c, UB\_c) with $m$ discrete values, therefore, the probability of $P\left(\mathcal{F}(c_1, q) \geq \mathcal{F}(c_i, q)\right)$ could be calculated using the maximum convolution of two probability distributions~\cite{rahman2015worker}.

\begin{lemma}
 {\tt ProbInd} takes $\Theta(M^2m)$ running time, where $M$ is the number of candidates, and $m$ is the number of discrete values of the pdf designating the score of candidate $c$.
\end{lemma}

\begin{proof}
(sketch.) Computing Max-convolution of any two pdfs takes $\Theta(m)$ times. Therefore, computing the probability of a candidate $c$ being the winner takes $\Theta(Mm)$ time.  {\tt ProbInd} repeats this process on each candidate and therefore takes $\Theta(M^2m)$ times.
\end{proof}

%\textcolor{red}{Please shorten every example of case B in 5-7 lines.}
\noindent \textbf {Case B - Dependence among candidates}
Algorithm {\tt ProbDep} is designed to identify the likely query answer from the candidate set, when the candidates have entities in common. It uses Equation~\ref{eq:joint} for that. Consider the running example $C=\{c_1,c_2,c_3\}$. Clearly, HNY exists in all $3$ candidates, MLN exists in both $c_1$ and $c_2$, as well as HYN exists in both $c_1$ and $c_3$. 

Two challenges immediately emerge: a. computational and storage bottleneck to compute the winning probability of a candidate using Equation~\ref{eq:joint}. b. disregarding the effect of common entities in the probability computation.

Using Equation \ref{eq:joint}, if there are $M$ candidates, the probability of a candidate being the winner is conditioned on as many as $M-1$ terms. Generally, given $M$ candidates, where the score of each candidate is a pdf with $m$ discrete values, there are $m^M$ possible combinations of scores. A naive implementation of {\tt ProbDep} needs to first identify which combinations satisfy the conditions expressed and computes the probability accordingly. This approach has a space complexity of $\mathcal{O}(m^M)$, which becomes prohibitively expensive as $M$ or $m$ increases.

%\textcolor{red}{I edited this - what you wrote was not correct. The notations were all messed up - we used $n$ for number of entities, cant abuse the notation. M is the number of candidates and m is the number of unique values.}

\subsubsection{Avoiding memory and computational bottleneck}
We therefore study storage and computational efficiency in designing {\tt ProbDep}. Algorithm~\ref{alg:dep} contains the pseudocode. Recall Equation~\ref{eq:jointexpand} and note that $P(c = c^*)$ could be calculated as a sequence of pairwise probabilities. Instead of storing all $\mathcal{O}(m^M)$ combinations, {\tt ProbDep} decomposes the overall computation and performs it in pairwise steps. Given the overall formula for $P(c_1 = c^*)$:

%\textcolor{orange}{comment from Sohrab: This is just the overall P(c\_1 = c*) formula. the final result obtain at each step i is of this form where it is written as follows: at step i, we compute ...}
%\textcolor{red}{I dont follow this at all...where is the conditioned expression here... }

%\[
%P(f(c_1) \geq f(c_2), f(c_1) \geq f(c_3), \dots, f(c_1) \geq f(c_{M-1}))
%\]
The algorithm runs iteratively and at step $i$, it computes: \[
P(\mathcal{F}(c_1, q) \geq \mathcal{F}(c_2, q), \mathcal{F}(c_1, q) \geq \mathcal{F}(c_3, q), \dots, \mathcal{F}(c_1, q) \geq \mathcal{F}(c_{i+1}, q))
\] by using probabilities calculated from step $(i - 1)$, which is of the following form.
\[P(\mathcal{F}(c_1, q) \geq \mathcal{F}(c_2, q), \mathcal{F}(c_1, q) \geq \mathcal{F}(c_3, q), \dots, \mathcal{F}(c_1, q) \geq \mathcal{F}(c_{i}, q))\]

%The process begins with:

%\[
%P(f(c_1) > f(c_2))
%\]

%and move one step at a time to:

%\[
%P(f(c_1) > f(c_2), \dots, f(c_1) > f(c_i))
%\]

%until reaching:

%\[
%P(f(c_1) > f(c_2), f(c_1) > f(c_3), \dots, f(c_1) > f(c_n))
%\]

\begin{comment}
    At each step \(i\), we only keep track of the scoring assignments from the previous step that satisfied:

\[
[f(c_1) > f(c_2), f(c_1) > f(c_3), \dots, f(c_1) > f(c_{i-1})]
\]

and use this information to compute:

\[
P(f(c_1) > f(c_2), f(c_1) > f(c_3), \dots, f(c_1) > f(c_i)).
\]
\end{comment}

%We maintain the feasible scoring assignments of ($c_1, c_i$) at step $i-1$ in a dictionary where the keys are the possible score values of $c_i$ and the corresponding value for each score of $c_i$ in this dictionary is the set of score values of $c$ that satisfy all the conditions at step $i - 1$. 
 Once step \(i\) is complete, the algorithm does not keep track of any past information from step \(i - 1\) anymore, and only maintains the results obtained from the latest step $i$. The process continues until all $M$ steps are complete.

Using our running example, $P(c_2 = c^*)$ is calculated in two steps: the algorithms first performs max-convolution to compute $P\left(\mathcal{F}(c_2, q) \geq \mathcal{F}(c_1, q)\right)$. In the next step, this computed information is then used to calculate $P\left(\mathcal{F}(c_2, q) \geq \mathcal{F}(c_3, q) \mid ( \mathcal{F}(c_2, q) \geq \mathcal{F}(c_1, q) \right)$. These two aforementioned probabilities are multiplied to produce $P(c_2 = c^*)$. 

\begin{algorithm} \label{compute_pdf_code}
\caption{Algorithm {\tt ProbDep}}\label{alg:dep}
\begin{algorithmic}[1]
    \State \textbf{Input:} Set of candidates $C = \{c_1, c_2, \dots, c_n\}$
    \State \textbf{Output:} Dictionary $all\_candidates\_probs$, where keys are candidates and values are their winning probabilities
    \State Initialize $all\_candidates\_probs \gets \{\}$
    
    \For{each candidate $c \in C$}
        \State Initialize $conditions \gets []$
        \State $prob\_recent \gets 1$
        
        \For{each candidate $c' \in C \setminus \{c\}$}
            \State Compute $P(\mathcal{F}(c, q) \geq \mathcal{F}(c', q) \mid conditions)$
            \State $prob\_recent \gets prob\_recent \times P(\mathcal{F}(c, q) \geq \mathcal{F}(c', q) \mid conditions)$
            \State Append $(c \geq c')$ to $conditions$
        \EndFor
        
        \State $all\_candidates\_probs[c] \gets prob\_recent$
    \EndFor
    
    \State \textbf{return} $all\_candidates\_probs$
\end{algorithmic}
\end{algorithm}

\textbf{Disregarding the effect of unknown common entities.} Given two candidates $c_i$ and $c_j$, {\tt ProbDep} needs to eliminate the effect of common unknown questions between $c_i$ and $c_j$ to compute the probability $P(\mathcal{F}(c_i, q) \geq \mathcal{F}(c_j, q))$.

 As an example, $c_1$ and $c_2$ have one unknown question in common, which is $R(HNY, q)$. {\tt ProbDep} assigns $R(HNY, q) = 0$ which means it has no effect on the score bounds of $c_1$ and $c_2$ particularly for computing $P(\mathcal{F}(c_2, q) \geq \mathcal{F}(c_1, q))$. Hence, assuming $R(HNY, q) = 0$, the new bounds are:

\[    \begin{array}{ll}
    c_1 & : \quad (3.5, 4.5), \\ 
    c_2 & : \quad (2.5, 3.5)
    \end{array}
\]

%While these new bounds are  obtained only to compute the exact \(P(\mathcal{F}(c_2, q) \geq \mathcal{F}(c_1, q))\), we use the original bounds of the two candidates for further computations. 

%These tightened temporary bounds eliminate some of the arrows in \autoref{fig:Example_3_6_a}(a), and change it to \autoref{fig:Example_3_6_a}(b). Hence, after applying the removal of common elements' effect, since there are 9 possible values and there is one edge, we have:

\begin{center}
\(P(\mathcal{F}(c_2, q) \geq \mathcal{F}(c_1, q)) = \frac{1}{9}\)    
\end{center}

\begin{comment}
    \begin{figure}[ht]
    \centering
    \includegraphics[width=0.4\textwidth]{figures/Example_3_6_a.jpg}
    \caption{Computing \(P(\mathcal{F}(c_2, q) \geq \mathcal{F}(c_1, q))\)}
    \label{fig:Example_3_6_b}
\end{figure}
\end{comment}

Having $P(\mathcal{F}(c_2,q) \geq \mathcal{F}(c_1,q))$  computed, the algorithm needs to next compute \(P(\mathcal{F}(c_2,q) \geq \mathcal{F}(c_3,q) \mid \mathcal{F}(c_2,q)\geq \mathcal{F}(c_1,q))\) to get the final probability of $P(c_2 = c^*)$. 
%\autoref{fig:Example_3_6_b}(a) Shows all the possible value combinations of $(c_1, c_2, c_3)$ such that $c_2 \geq c_3$ given $c_2 \geq c_1$ as the condition. The way we add edges from middle partition to the right partition is by considering the most recent information that we have stored which is the corresponding dictionary of the middle partition which stores all possible score assignments that satisfy $c_2 \geq c_1$ as the condition for the probability that now we want to compute in our current step. Hence, edges going from $c_2$ to $c_3$ are chosen from values of this dictionary which already satisfy $c_2 \geq c_1$. 

However, {\tt ProbDep} now eliminates the effect of common unknown questions between $c_2$ and $c_3$, which is $R(HNY, q)$. Hence, assuming it is zero, the new bounds are:

\[\begin{array}{ll}
    c_2 & : \quad (2.5, 3.5), \\ 
    c_3 & : \quad (3.0, 4.0)
    \end{array}
\]
This gives \(P(\mathcal{F}(c_2,q) \geq \mathcal{F}(c_3,q) \mid \mathcal{F}(c_2,q)\geq \mathcal{F}(c_1,q))\) = $\frac{2}{45}$

\begin{comment}
    \begin{figure}[ht]
    \centering
    \includegraphics[width=0.5\textwidth]{figures/Example_3_6_b.jpg}
    \caption{\(\mathbf{P(c_2 \geq c_3 \mid c_2 \geq c_1)}\)}
    \label{fig:Example_3_6_a}
\end{figure}

\end{comment}

Finally, these are multiplied to obtain $P(c_2 = c^*)$. Similarly, the probabilities \( P(c_1 = c^*) \) and \( P(c_3 = c^*) \) can be computed.

\begin{comment}
If we observe the bounds of these three candidates without calculating exact probabilities, we can estimate that \( c_2 \) is less likely to achieve the highest score compared to the others. While unknown questions could potentially impact the outcome, common unknown questions between candidates will affect all candidates similarly and are not game-changers in this case. Therefore, when we exclude their influence, \( P(c_2 = c^*) \) becomes even lower. This is reasonable, as unknown common questions cannot help \( c_2 \) exceed its score relative to the others.

\begin{figure}[ht]
    \centering
    \includegraphics[width=0.4\textwidth]{figures/PDF.jpg}
    \caption{Probability Distribution of the winner candidate}
    \label{fig:PDF}
\end{figure}

One challenge in calculating these probabilities involve quantifying the effect of common elements and disregarding their effect in the final probability value. \textcolor{red}{As an example, }
\end{comment}

It could be shown that Algorithm~{\tt ProbDep} will produce
\[
P(c = c^*) =
\begin{cases} 
0.75 & \text{ } c = c_1, \\
0.24 & \text{ } c = c_2, \\
0.01 & \text{ } c = c_3.
\end{cases}
\]

 \begin{lemma}
 {\tt ProbDep} takes $\Theta(M^2m^2)$ running time and $O(m^2)$ space.
 \end{lemma}

\begin{proof}
(sketch.) Computing Max-convolution of two pdfs takes $\Theta(m^2)$ times. Therefore, computing the probability of a candidate $c$ being the winner takes $\Theta(Mm^2)$ time.  {\tt ProbDep} repeats this process on each candidate and therefore takes $\Theta(M^2m^2)$ times.
\end{proof}

\begin{comment}
  \subsubsection{Avoiding Memory Bottleneck} 
Through our stepwise process for computing winning probability for each candidate $c$, we iteratively compute probabilities of the form $P(c > c_i | conditions)$. For computing each term of the mentioned form, instead of maintaining a n-dimensional table to consider all possible value combinations for all candidates, we utilize memory usage by only storing the useful information which is a dictionary that its keys are the possible score values of $c_{i-1}$, and its values are the set of score values for $c$ that satisfy the conditions with respect to that $c_{i-1}$ value key. Hence, we only need to store $\mathcal{O}(m^2)$ instead of a huge table that requires $\mathcal{O}(m^n)$ space.

\label{uncertainty}  
\end{comment}

%Define entropy

% design an example of already asked questions, compute entropy, select two question, compare their effect on entropy, why one is better than another

% hence, we propose algorithm to find a question which most likey reduces entropy maximally compared to others
\subsection{Determining the Next Question}\label{subsec:nextquestion}
In \autoref{winning_probability}, we discussed algorithms for the probabilistic model for finding the answer of the query. Formally speaking, the query answer is a random variable with $M$ possible outcomes (each per candidate), and their probability could be computed using {\tt ProbDep} or {\tt ProbInd}. We use {\em entropy} as a measure of the uncertainty associated with this random variable as follows:

\begin{align}\label{eq:entropy}
H(c^*) = - \sum_{i=1}^{n} p(c_i = c^*) \log p(c_i = c^*)
\end{align}

The entropy of $c_1, c_2, c_3$ in the running example is  $ H(c^*) = - \left( 0.75 \log(0.75) + 0.01 \log(0.01) + 0.24 \log(0.24) \right) = 0.604$
  
The next best question should therefore be the one that minimizes entropy (ideally makes it $0$). However, the challenge is to select this question from $Q_U$ without any further assumption about the response received from LLM. Algorithm {\tt EntrRed} is designed for this task (Algorithm~\ref{alg:entred} has the pseudocode).

 This algorithm leverages {\tt ProbDep} or {\tt ProbInd} and first identifies the candidate (let $c^+$ be that candidate) that has the  highest probability to be the winner. Given $Q_U$, it first narrows down to a smaller subset $Q_{U'}$ that involve  $c^+$. It leverages a subroutine {\tt QEF} (Subroutine~\ref{alg:qef}) to quantify the effect of every question $Q \in Q_{U'}$ and then selects $Q \in Q_{U'}$ that has the maximum score associated.

 %\textcolor{orange}{No intuition - does not explain why we did what we did. Needs to be rewritten .}
Subroutine {\tt QEF} works as follows: For every $Q \in Q_{U'}$, let $C_Q$ represent the subset of all candidates whose scores are influenced by $Q$, and $C'_Q$ be the other candidates. The score assigned to the question $Q$ is the sum of absolute difference between winning probability of candidates in $C_Q$ and winning probability of candidates in $C'_Q$. 

Intuitively, consider a pair of candidates $c_Q \in C_Q$ and $c'_Q \in C'_Q$. if $|P(c_Q = c^*) - P(c'_Q = c^*)|$ is high, it means that among $c_Q$ and $c'_Q$, one has a much higher winning probability compared to the other which means the score bounds of $c_Q$ and $c'_Q$ have a small overlap, and asking $Q$ as a question that influences the score of only one of the two candidates, can potentially diverge their bounds, and make one of the two candidates certainly better than the other, which will prune out one of the candidates. On the other hand, if $Q$ influences both candidates, asking it will affect score bounds of both candidates in the same way, and will not result in pruning one of them. This is why for a given $Q$, we divide candidates by two groups based on if they are influenced by $Q$ or not, and then consider the winning probability difference between candidates from the two groups. The higher each pairwise difference in winning probability is, the more valuable the $Q$ is since it becomes more likely to prune out candidates from the two groups by asking $Q$. 

\begin{algorithm}
\caption{Algorithm {\tt EntrRed}}\label{alg:entred}
\begin{algorithmic}[1]
    \State \textbf{Input:} Set of unknown questions $\mathcal{Q_U}$, Probability Density Function (PDF) for candidates
    \State \textbf{Output:} A single question $Q^*$ to be asked next
    \State Initialize $Q^* \gets \text{null}$
    \State Initialize $max\_prob \gets -\infty$
    \State Initialize $c^+ \gets \text{null}$  \Comment{Candidate with highest probability}
    
    \For{each candidate $c$}
        \State $prob(c) \gets \text{PDF}(c)$ \Comment{Obtain probability for candidate $c$}
        \If{$prob(c) > max\_prob$}
            \State $max\_prob \gets prob(c)$
            \State $c^+ \gets c$  \Comment{Update best candidate}
        \EndIf
    \EndFor

    \State Let $Q_{U'}$ be the set of questions contributing to the score of $c^+$
    \State Initialize $max\_score \gets -\infty$
    \For{each question $Q \in Q_{U'}$}
        \State $score(Q) \gets \text{Evaluate\_Question}(Q, C_Q)$ \Comment{$C_Q$: candidates influenced by $Q$}
        \If{$score(Q) > max\_score$}
            \State $max\_score \gets score(Q)$
            \State $Q^* \gets Q$  \Comment{Update best question}
        \EndIf
    \EndFor
    
    \State \textbf{return} $Q^*$
\end{algorithmic}
\end{algorithm}

\begin{algorithm}
\caption{Subroutine QEF}\label{alg:qef}
\begin{algorithmic}[1]
    \Function{evaluate\_question}{Q, $C_Q$}
        \State Initialize $score(Q) \gets 0$
        \For{each candidate $c \in C_Q$}
            \State Initialize $prob\_diff\_sum(c) \gets 0$
            \For{each candidate $c' \notin C_Q$}
                \State $prob\_diff(c,c') \gets \left| P(c = c^*) - P(c' = c^*) \right|$
                \State $prob\_diff\_sum(c) \gets prob\_diff\_sum(c) + prob\_diff(c,c')$
            \EndFor
            \State $score(Q) \gets score(Q) + prob\_diff\_sum(c)$
        \EndFor
        \State \Return $score(Q)$
    \EndFunction
\end{algorithmic}
\end{algorithm}

As we ask more questions, we expect the entropy to decrease. Finally, when the winner candidate is obtained, the entropy goes down to zero. 

Using our running example, {\tt EntrRed} narrows down to $Q_{U'}$ that are pertinent to $c_1$ only. Among unknown questions of $c_1$, $Q_1= R(HNY)$ and $Q_2=D(MLN, HYN)$ are the two possibilities. {\tt QEF} is invoked for both of them.
  $Q_1$ = R(HNY, q): We need to compute the sum of winning probability difference between candidates from $C_{Q1}$ with any other candidates. Since $C_{Q1} = \{c_1, c_2, c_3\}$, there will be no remaining other candidates for computing overlap with these 3. Hence, QEF(Q1) = 0
$Q_2$ = D(MLN, HYN): Since $C_{Q2} = {c_1}$, $QEF(Q2) = |p(c_1) - p(c_2)| + |p(c_1) - p(c_3)|=1.25$. $Q_2$ is therefore asked.
Let's assume that the oracle (LLM) returns $D(MLN, HYN) = 1.0$. Then, it could be shown, that
\[
P(c = c^*) =
\begin{cases} 
1.0 & \text{ } c = c_1, \\
0.0 & \text{ } c = c_2, \\
0.0 & \text{ } c = c_3.
\end{cases}
\]
Clearly, at this point the entropy of the random variable representing different outcomes in Equation~\ref{eq:entropy} is $0$. $c_1$ is thus returned as the answer of the query.

\begin{lemma}
  {\tt EntrRed} takes $O(|Q_c|M^2)$ time to run.
\end{lemma}

\begin{proof}
(Sketch.)
Subroutine {\tt QEF} divides all $M$ candidates in two groups and computes pairwise differences in winning probability  among them. Hence, it takes $(M^2)$ time. {\tt EntrRed} repeats this process on each question influencing the score of the most possible winner candidate and therefore takes $O(|Q_c|M^2)$ time.
\end{proof}

\subsection{Response Processing}\label{subsec:resp}
For a given question $Q$, a discrete response \(r \in [MIN,MAX] \), within the specified range \([MIN, MAX]\), from a single oracle could be a normalized floating point number representing a score value (e.g. 0.7). Upon receiving $r$, the score bounds of each candidate $c$ that contains $Q$ are updated as follows: for the lower bound, substitute MIN  by $r$, and for the upper bound substitute MAX by $r$. 
\section{Experimental Evaluation}\label{sec:exp}
We present the effectiveness of our framework using two key metrics: 
\textbf{M1} Quality of the solution, which assesses the framework's ability to reduce costs (i.e., number of LLM calls) compared to alternative approaches; and 
\textbf{M2} Scalability of each component, which analyzes the time trends as the candidate space grows for each component. 

\subsection{Experimental Setup}\label{sec:exp_setup}

\noindent \textbf{Experiment Settings.} 
Algorithms are implemented in Python 3.11.1, utilizing GPT-4o mini as the oracle/expert in our framework. Experiments are conducted on Wulver (NJIT's HPC cluster using 6 nodes) with a 2.45 GHz AMD EPYC 7753 processor and 512 GB RAM (per node). The code and data are publicly available.\footnote{\it {\it \href{https://github.com/sohrabnamazinia/Personalized-Top-K-Set-Queries}{https://github.com/sohrabnamazinia/Personalized-Top-K-Set-Queries}}} Results are averaged over 10 runs.

\subsubsection{Applications}\label{sec:exp_apps}

We have three use cases: hotels, movies, and Yelp businesses. In all cases, the input comprises a user's query, \( k \), and personalized definitions for each scoring function construct (relevance and diversity). Each LLM prompt requests either the relevance (or diversity) score of an entity (or pair of entities). We use six scoring functions (two per use case) with distinct relevance and diversity definitions as shown in Table \ref{tab:Scoring_functions}.  Prompts are carefully designed to include all necessary unstructured data, i.e., the user's query, relevance and diversity definitions, and unstructured data associated with the related entity (or entities). We use the LangChain framework~\footnote{\url{https://www.langchain.com/}} to ensure that the LLM returns scores in our desired normalized floating-point format. 

\noindent \textbf{1) Top-k Hotels.} 
This is a large dataset \cite{arnab_das_2024} of hotels, containing detailed descriptions of 719,218 hotels sourced from websites, travel agencies, and review platforms. These descriptions serve as the unstructured data associated with hotel entities that we use in our prompts. Two scoring functions with distinct relevance and diversity definitions are used, \(\mathcal{F}_1\) and \(\mathcal{F}_2\) defined in Table \ref{tab:Scoring_functions}.

\begin{comment}
  \begin{itemize}
  \item \(\mathcal{F}_1\): 
    \begin{itemize}
      \item Relevance definition: "Hotel rating"
      \item Diversity definition: "Physical distance between hotels"
    \end{itemize}
  \item \(\mathcal{F}_2\):
    \begin{itemize}
      \item Relevance definition: "Proximity to city center"
      \item Diversity definition: "Star rating"
    \end{itemize}
\end{itemize}  
\end{comment}

\noindent \textbf{2) Top-k Movies.} 
This is the Wikipedia Movie Plots\footnote{\it {\it \href{https://www.kaggle.com/datasets/jrobischon/wikipedia-movie-plots}{https://www.kaggle.com/datasets/jrobischon/wikipedia-movie-plots}}}, which includes 33,869 movies with  plot descriptions from Wikipedia pages. These plots form the unstructured data associated with movie entities in our prompts. Two scoring functions with distinct relevance and diversity definitions are used, \(\mathcal{F}_3\) and \(\mathcal{F}_4\) defined in Table \ref{tab:Scoring_functions}.

\begin{comment}
    \begin{itemize}
  \item \(\mathcal{F}_3\): 
    \begin{itemize}
      \item Relevance definition: "Brief plot"
      \item Diversity definition: "Production year"
    \end{itemize}
  \item \(\mathcal{F}_4\):
    \begin{itemize}
      \item Relevance definition: "Popularity"
      \item Diversity definition: "Genres \& movie eras"
    \end{itemize}
\end{itemize}
\end{comment}

\noindent \textbf{3) Top-k businesses in Yelp dataset.} This is the Yelp dataset, which contains various attributes for 150,346 businesses (e.g., restaurants), including reviews and images. We used a bundle of reviews and images as the associated unstructured data for each business entity in our prompts. LangChain is used for image processing with LLMs. Two scoring functions with distinct relevance and diversity definitions are used, \(\mathcal{F}_5\) and \(\mathcal{F}_6\) defined in Table \ref{tab:Scoring_functions}.

\begin{comment}
   \begin{itemize}
  \item \(\mathcal{F}_5\): 
    \begin{itemize}
      \item Relevance definition: "Location near New York"
      \item Diversity definition: "Cost variety"
    \end{itemize}
  \item \(\mathcal{F}_6\):
    \begin{itemize}
      \item Relevance definition: "Cuisine type"
      \item Diversity definition: "Varied operating hours"
    \end{itemize}
\end{itemize} 
\end{comment}

\subsubsection{Implemented Algorithms}\label{sec:exp_algs}
\ \\
There does not exist any related work that studies the problem that we do. However, we still design baselines that are appropriate.\\
\noindent {\bf ***} \texttt{Baseline.} This approach makes all LLM calls to calculate the exact score for all candidates without maintaining bounds. It then selects the highest scoring candidate.\\
\noindent {\bf ***} \texttt{EntrRed} using {\tt ProbDep.} The original version of our  framework, which does not assume independence among candidates and accounts for their dependencies in the computation.\\
\noindent {\bf ***} {\tt EntrRed} using {\tt ProbInd.} A time-efficient variation of our  framework that assumes independence among candidates, enabling more efficient computation of the probability distribution function (pdf).\\
\noindent {\bf ***} {\tt Random}. This approach maintains score bounds. However, instead of computing a pdf and determining the next question based on that, it selects the next best question randomly.

\subsubsection{Measures}\label{sec:exp_measures}
\ \\ 
%We evaluated the effectiveness and efficiency of our  framework using the following metrics:
\noindent \textbf{Number of LLM Calls.} The primary objective of our framework is to identify the top-k set with the highest score while minimizing the cost of using LLMs as oracles/experts. The number of LLM calls serves as a key indicator of cost. %, particularly in large-scale data management applications where LLM calls can be expensive. 
We hence compared the total number of LLM calls for {\tt Random} vs. {\tt EntrRed} using {\tt ProbInd}. We also compared {\tt EntrRed} using {\tt ProbDep} vs. {\tt EntrRed} using {\tt ProbInd} to investigate if computing probabilities considering dependence using {\tt EntrRed} using {\tt ProbDep} can result in reducing cost more effectively compared to {\tt EntrRed} using {\tt ProbInd}. 

\noindent \textbf{Recall.} We evaluate recall of the top-$k$ produced by our designed solutions wrt baseline. 

\noindent \textbf{Time Taken for Each Component.} We measured the total time taken by the different algorithms and the time taken for each of the four components of our framework. 

\subsection{Results}\label{sec:results}
The top-$k$ set produced by our proposed algorithms achieve 100\% recall always, as expected.  
\subsubsection{\noindent \textbf{Cost Experiments.}}
Given a fixed number of candidates $n$, we vary \( k \) and report the cost (\# LLM calls) of finding the top-k set. Since Baseline performs all  calls, it is omitted from this experiment. 

\begin{figure*}[!htbp]
    \centering
    \subfigure[Hotels - Scoring function \( \mathcal{F}_1 \)]{
        \includegraphics[width=0.32\textwidth,height=0.18\textheight]{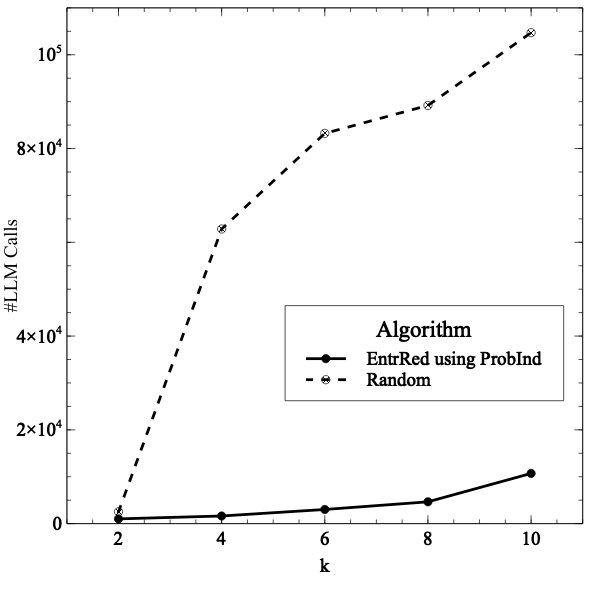}
    }\hfill
    \subfigure[Hotels - Scoring function \( \mathcal{F}_2 \)]{
        \includegraphics[width=0.32\textwidth,height=0.18\textheight]{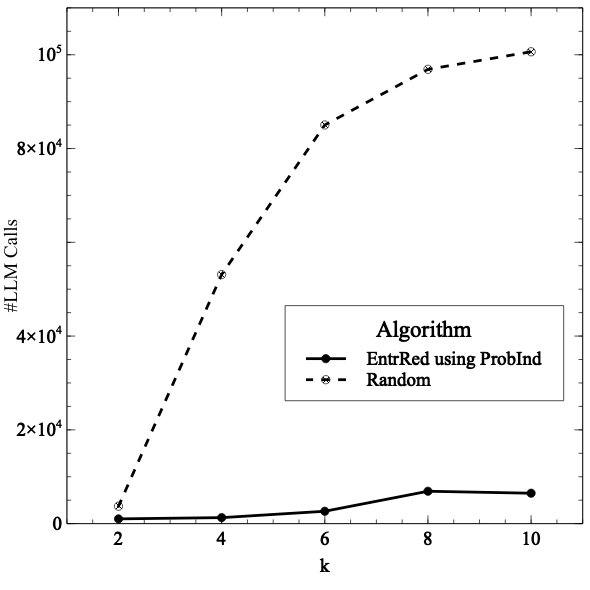}
    }\hfill
    \subfigure[Movies - Scoring function \( \mathcal{F}_3 \)]{\includegraphics[width=0.32\textwidth,height=0.18\textheight]{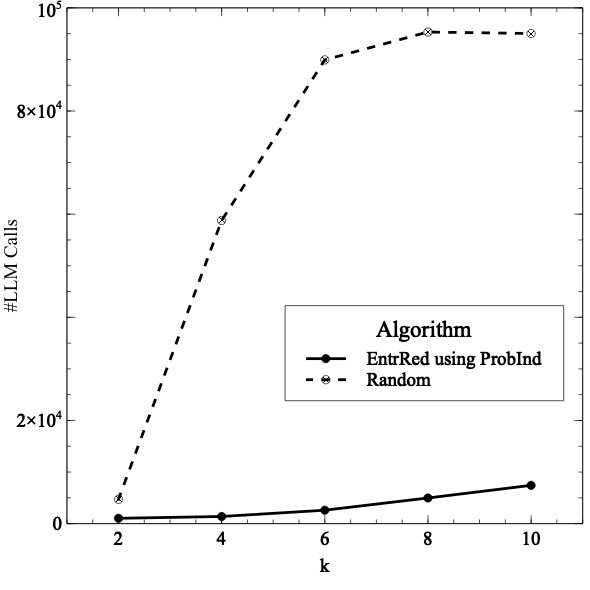}
    }\vfill
    \subfigure[Movies - Scoring function \( \mathcal{F}_4 \)]{
        \includegraphics[width=0.32\textwidth,height=0.18\textheight]{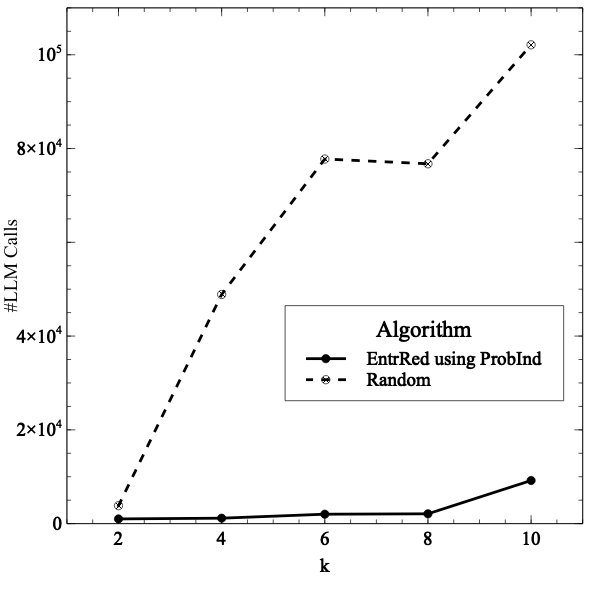}
    }\hfill
    \subfigure[Businesses - Scoring function \( \mathcal{F}_5 \)]{
        \includegraphics[width=0.32\textwidth,height=0.18\textheight]{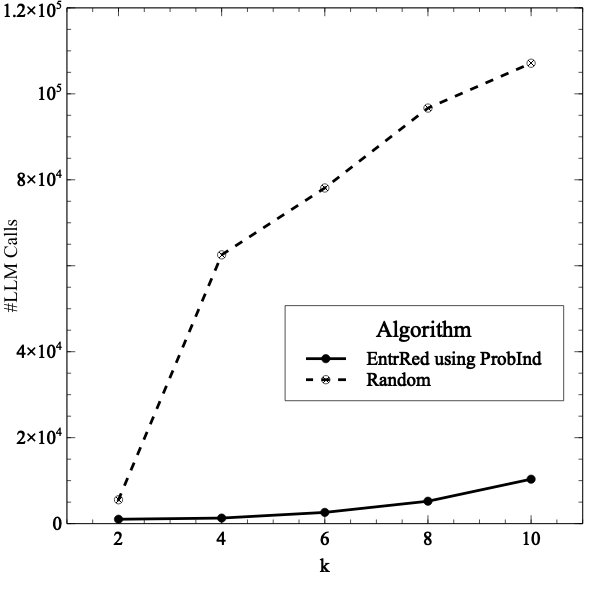}
    }\hfill
    \subfigure[Businesses - Scoring function \( \mathcal{F}_6 \)]{
        \includegraphics[width=0.32\textwidth,height=0.18\textheight]{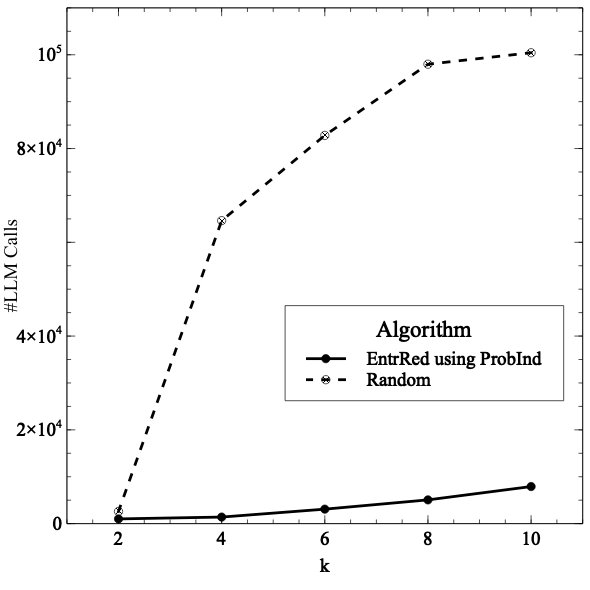}
    }
    \caption{\small \#LLM calls varying \( k \) - {\tt EntrRed} using {\tt ProbInd} vs. {\tt Random}}
    \label{fig:quality_ind_random}
\vspace{-0.1in}
\end{figure*}
%\vspace{-0.05in}
1) {\tt EntrRed} using {\tt ProbInd} vs. {\tt Random}:
Figure \ref{fig:quality_ind_random} compares the cost between our  framework with {\tt Random} across the 6 scoring functions defined over the 3 datasets, as detailed in the previous subsection \ref{sec:exp_setup}. These figures demonstrate two key observations:
\begin{itemize}
    \item Our method for determining the next question significantly reduces  the cost (by an order of magnitude) compared to the {\tt Random} algorithm. 
    \item As expected, the cost increases with \( k \) for both the solutions, because a larger \( k \) involves more questions contributing to each candidate's score. However, the cost increases sublinearly for ours with increasing $k$.
\end{itemize}

2) {\tt EntrRed} using {\tt ProbDep} vs. {\tt EntrRed} using {\tt ProbInd}: Here, we focus on comparing the two variants of our  proposed probabilistic models discussed in Section~\ref{winning_probability}. Figure \ref{fig:quality_ind_dep} compares the costs of {\tt EntrRed} using {\tt ProbInd} versus {\tt EntrRed} using {\tt ProbDep}. For brevity, we only present a subset of results that are representative (additional results could be found in our technical report~\cite{tr}).
\begin{itemize}
    \item  {\tt ProbDep} outperforms {\tt ProbInd} in terms of cost for {\tt EntrRed}. This is expected, since unlike in {\tt ProbInd},  {\tt ProbDep} factors in dependence among candidates in calculating the winning probability of each candidate, requiring smaller number of oracle calls in the end.
    \item However, the difference in cost between these two variants is minor, validating that either variant is highly effective for cost reduction compared to {\tt Random}.
\end{itemize}

\subsubsection{\noindent \textbf{Scalability Experiments.}} We present the scalability of our framework by measuring running time of the designed algorithms of each of the four tasks, as well as computing the total running time needed for solving all four tasks. Since {\tt Random} is very poor qualitatively, it is omitted from these experiments.

1) \textbf{Total time taken}: In this exp, we vary the size of the candidate set and measure overall running time of the framework. Figure \ref{fig:scalibility_total_time} shows the total time taken by our framework with {\tt EntrRed} using {\tt ProbDep} vs. {\tt EntrRed} using {\tt ProbInd}. The independence assumption in {\tt EntrRed} using {\tt ProbInd} makes the solution significantly lean in computation time compared to {\tt EntrRed} using {\tt ProbDep} - as expected, the former exhibits an order of magnitude speed up compared to the latter. 

\begin{comment}
2) \textbf{Identifying candidates}: Figures \ref{fig:scalibility_init_candidates_results} and \ref{fig:scalibility_init_candidates_last_two} depict the total time taken for identifying the candidates set. This time is negligible and scales well with an increasing number of candidates.    
\end{comment}

2) \textbf{Computing bounds}: Figure \ref{fig:scalibility_computing_bounds} shows the time for computing bounds in our framework (additional results could be found in our tech report~\cite{tr}). These results corroborate our theoretical analysis - computing bounds is a lightweight task and scales very well. The time takes in negligible compared to the total time.

3) \textbf{Probabilistic model}: This is the most computationally expensive component of the framework. Figure \ref{fig:scalibility_compute_pdf} illustrates (refer to~\cite{tr} for additional results) the time for computing the probabilistic models. Additional results could be found in our technical report~\cite{tr}. Computing probabilities using {\tt ProbDep} is significantly more expensive compared to {\tt ProbInd}, where we assume candidates are independent. This results in a challenge for scalability when the number of candidates grows, whereas {\tt ProbInd} remains highly scalable. Indeed, computing the probabilistic model using {\tt ProbInd} is an order of magnitude faster than its counterpart.

4) \textbf{Determining the next question}: Figure \ref{fig:scalibility_determine_next_q} shows the time for the next best question of {\tt EntrRed} using {\tt ProbInd}. Additional results could be found in our technical report~\cite{tr}. As expected, it is highly efficient and scales well when increasing the number of candidates.

5) \textbf{LLM response}: Figure \ref{fig:scalibility_llm_response} depicts the time taken by the LLM to respond to a question by our proposed framework. {\tt Baseline} makes considerably more LLM calls, leading to a significant increase in processing time. These results reinforce our motivation - calling an LLM is expensive and hence minimizing that cost is important.

\subsection{Summary of Results}
Our experimental evaluation reveals two key findings: 1. As expected, our proposed framework consistently returns the exact top-$k$ results. Additionally, it is highly effective in minimizing the number of LLM calls, achieving a reduction by an order of magnitude compared to baseline methods. The framework demonstrates exceptional generalizability, performing well across various large-scale applications that involve multimodal data. 2. The scalability investigations conducted in this study prove effective. As theoretically analyzed, computing the probabilistic model is the most time-consuming component of the framework. The improvements we suggested in our scalable solution, {\b EntrRed} using {\tt ProbInd}, are both effective and efficient.

\begin{comment}
    In a nutshell, Here is the order of time taken for each component of EntrRedInd and EntrRedDep:
\[
\text{Computing PDF} \gg \text{LLM Response} > \text{Determine next question}
\]
\[
> \text{Update bounds} > \text{Init candidates}.
\]
\end{comment}

\begin{figure*}[!htbp]
    \centering
    \subfigure[Hotels - Scoring function \( \mathcal{F}_1 \)]{
        \includegraphics[width=0.32\textwidth,height=0.18\textheight]{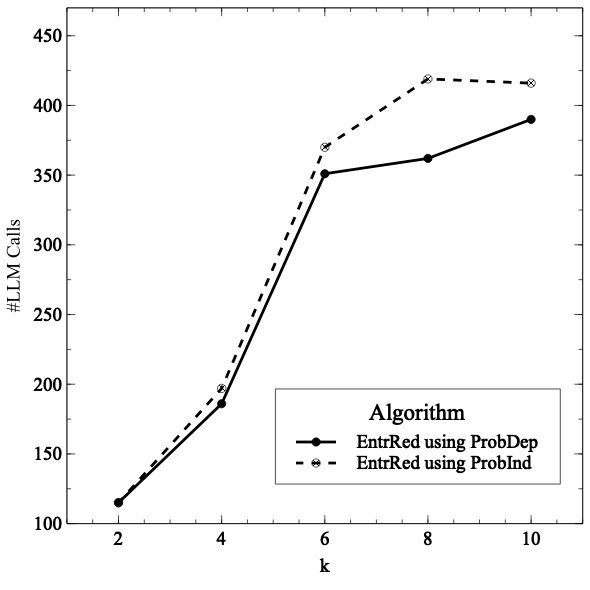}
    }\hfill
    \subfigure[Movies - Scoring function \( \mathcal{F}_3 \)]{
        \includegraphics[width=0.32\textwidth,height=0.18\textheight]{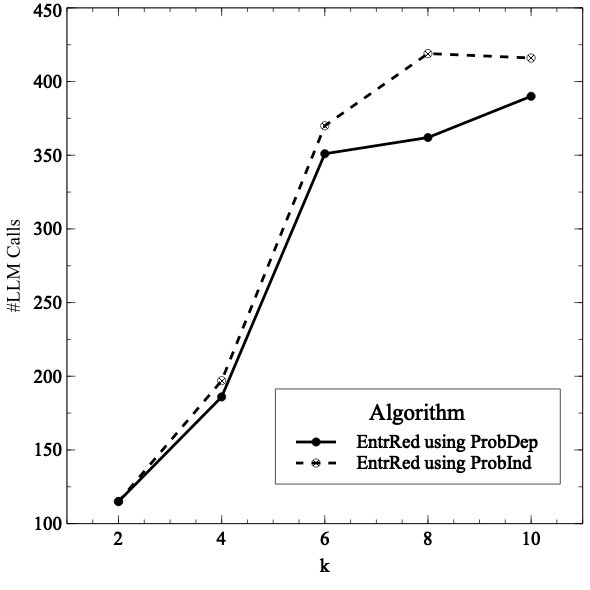}
        }\hfill
    \subfigure[Businesses - Scoring function \( \mathcal{F}_5 \)]{
        \includegraphics[width=0.32\textwidth,height=0.18\textheight]{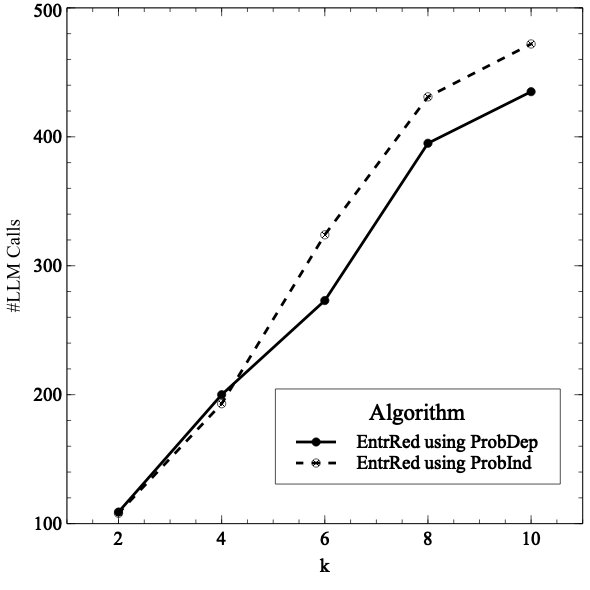}
    }
    \caption{\small \#LLM calls varying \( k \) - {\tt EntrRed} using {\tt ProbDep} vs. {\tt EntrRed} using {\tt ProbInd}}
    \label{fig:quality_ind_dep}
     \vspace{-0.1in}
\end{figure*}

\begin{table}[!htbp]
\centering
\begin{tabular}{|c|c|c|}
\hline
\textbf{Scoring} & \textbf{Relevance} & \textbf{Diversity} \\ 
\textbf{Function} & \textbf{Construct} & \textbf{Construct} \\ \hline
\(\mathcal{F}_1\) & Hotel rating & Physical distance \\ \hline
\(\mathcal{F}_2\) & Proximity to city center & Star rating \\ \hline
\(\mathcal{F}_3\) & Brief plot & Production year \\ \hline
\(\mathcal{F}_4\) & Popularity & Genres \& movie eras \\ \hline
\(\mathcal{F}_5\) & Location near New York & Cost variety \\ \hline
\(\mathcal{F}_6\) & Cuisine type & Varied operating hours \\ \hline
\end{tabular}
\caption{\small Personalized scoring functions used in experiments}\label{tab:Scoring_functions}
\end{table}
\vspace{-0.3in}

\begin{figure*}[!htbp]
    \centering
    \subfigure[Hotels - Scoring function \( \mathcal{F}_1 \)]{
        \includegraphics[width=0.32\textwidth,height=0.18\textheight]{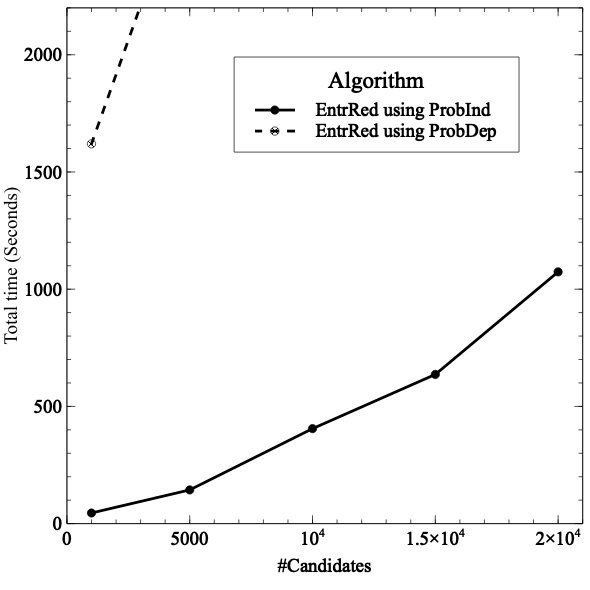}
    }\hfill
    \subfigure[Hotels - Scoring function \( \mathcal{F}_2 \)]{
        \includegraphics[width=0.32\textwidth,height=0.18\textheight]{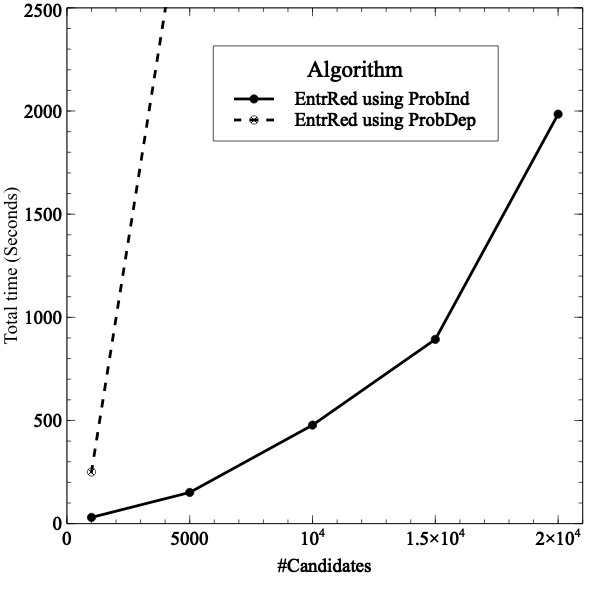}
    }\hfill
    \subfigure[Movies - Scoring function \( \mathcal{F}_3 \)]{
        \includegraphics[width=0.32\textwidth,height=0.18\textheight]{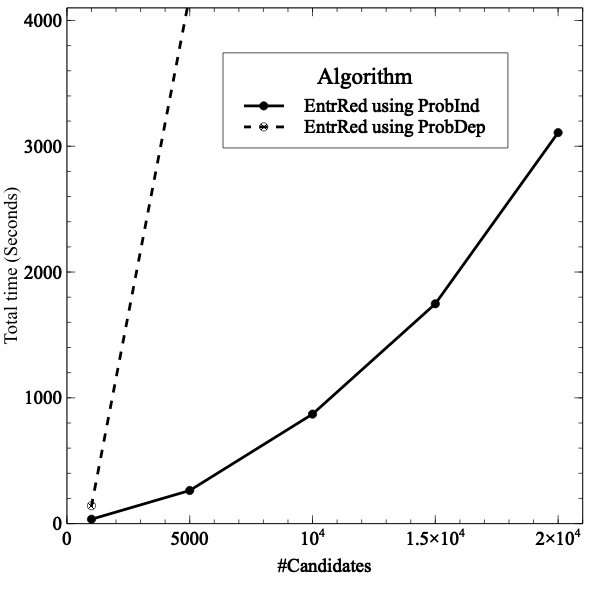}
    }\vfill
    \subfigure[Movies - Scoring function \( \mathcal{F}_4 \)]{
        \includegraphics[width=0.32\textwidth,height=0.18\textheight]{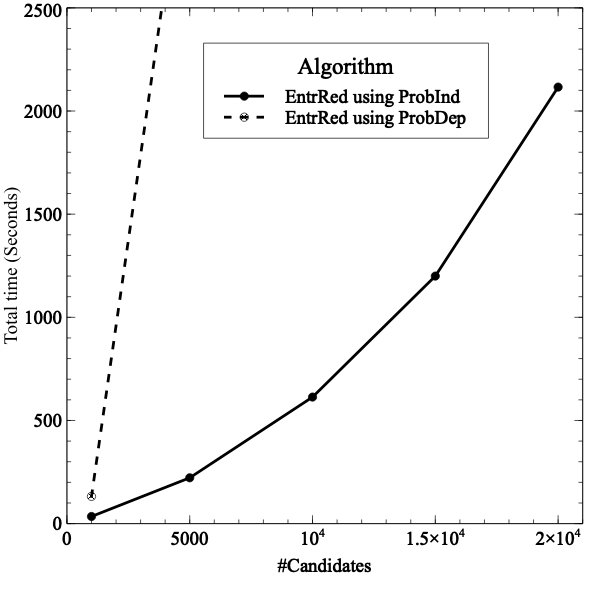}
    }\hfill
    \subfigure[Businesses - Scoring function \( \mathcal{F}_5 \)]{
        \includegraphics[width=0.32\textwidth,height=0.18\textheight]{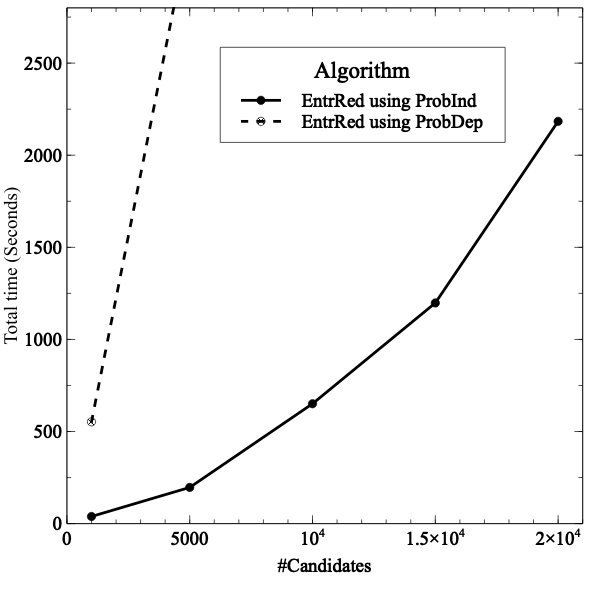}
    }\hfill
    \subfigure[Businesses - Scoring function \( \mathcal{F}_6 \)]{
        \includegraphics[width=0.32\textwidth,height=0.18\textheight]{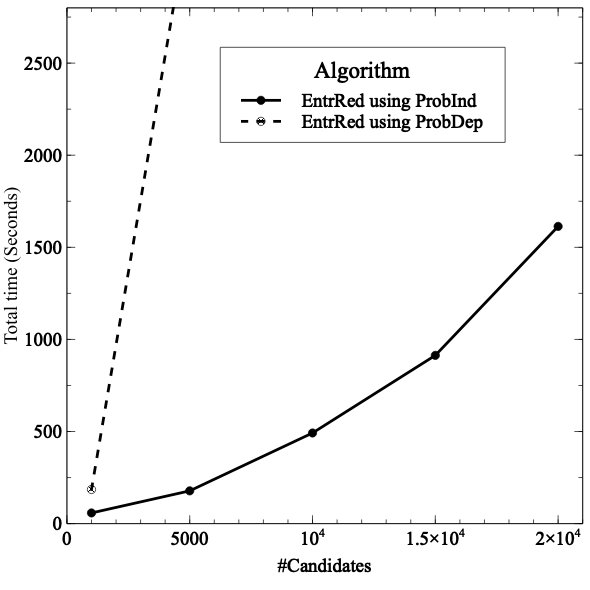}
        }
    \caption{\small Total time taken for {\tt EntrRed} using {\tt ProbDep} vs. {\tt EntrRed} using {\tt ProbInd} varying \#candidates.}
    \label{fig:scalibility_total_time}
    \vspace{-0.1in}
\end{figure*}

\begin{figure*}[!htbp]
    \centering
    \subfigure[Hotels - Scoring function \( \mathcal{F}_1 \)]{
        \includegraphics[width=0.32\textwidth,height=0.18\textheight]{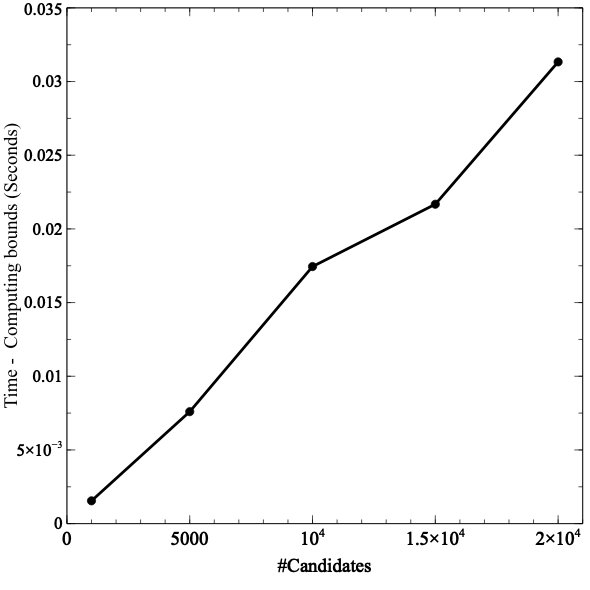}
    }\hfill
    \subfigure[Movies - Scoring function \( \mathcal{F}_3 \)]{
        \includegraphics[width=0.32\textwidth,height=0.18\textheight]{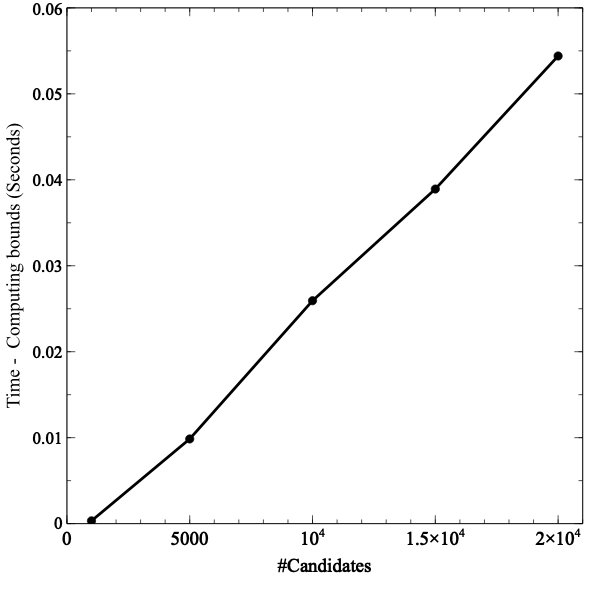}
    }\hfill
    \subfigure[Businesses - Scoring function \( \mathcal{F}_5 \)]{
        \includegraphics[width=0.32\textwidth,height=0.18\textheight]{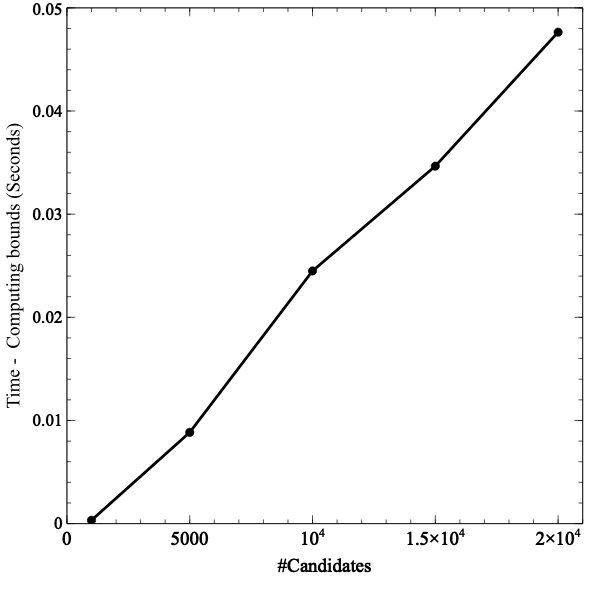}
    }
    \caption{\small Time taken for computing bounds varying \#candidates.}
    \label{fig:scalibility_computing_bounds}
     \vspace{-0.1in}
\end{figure*}

\begin{figure*}[!htbp]
    \centering
    \subfigure[Hotels - Scoring function \(\mathcal{F}_1\)]{
        \includegraphics[width=0.32\textwidth,height=0.18\textheight]{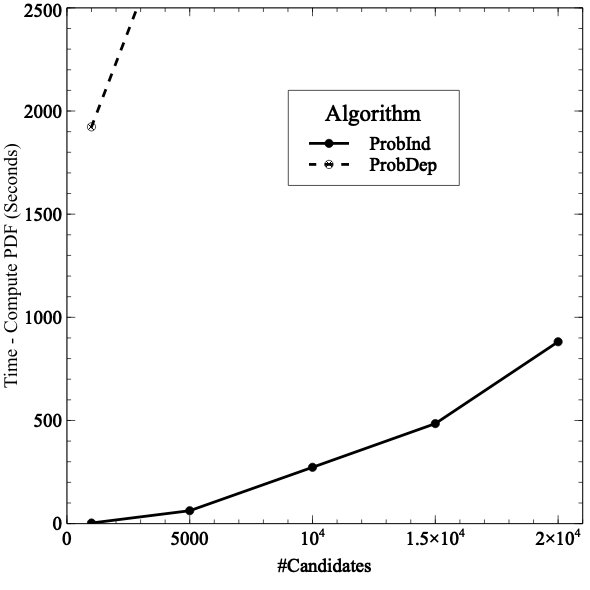}
    }\hfill
    \subfigure[Businesses - Scoring function \(\mathcal{F}_3\)]{
        \includegraphics[width=0.32\textwidth,height=0.18\textheight]{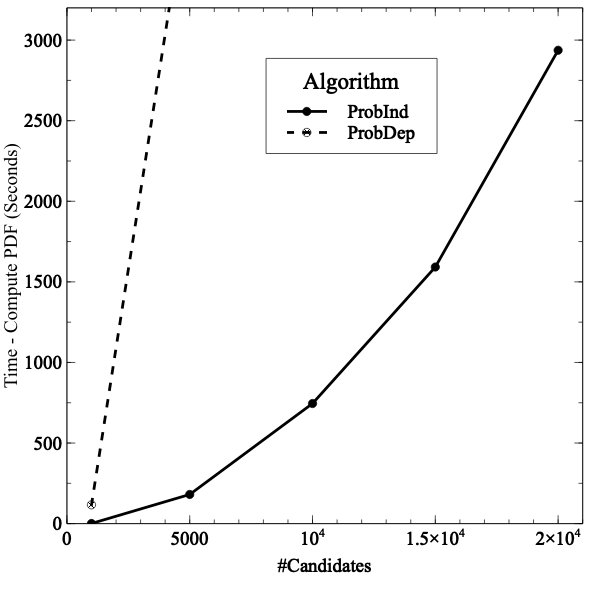}
    }\hfill
    \subfigure[Businesses - Scoring function \(\mathcal{F}_5\)]{
        \includegraphics[width=0.32\textwidth,height=0.18\textheight]{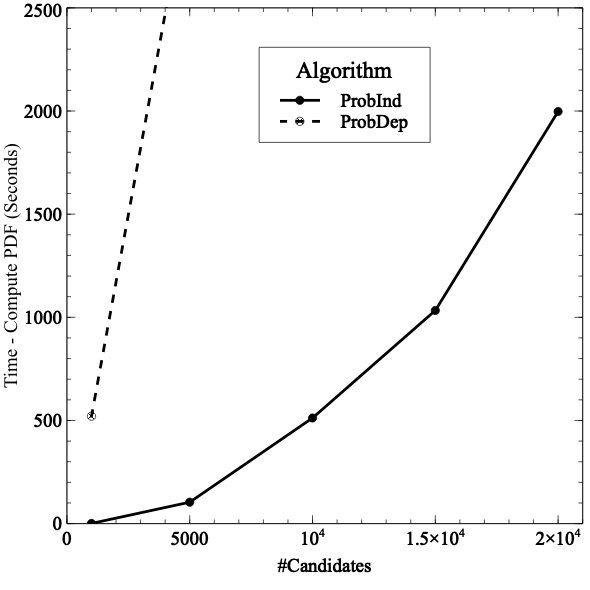}
    }
    \caption{\small Time taken for computing probabilistic model {\tt ProbInd} vs. {\tt ProbDep} varying \#candidates.}
    \label{fig:scalibility_compute_pdf}
    \vspace{-0.1in}
\end{figure*}

\begin{figure*}[!htbp]
    \centering
    \subfigure[Hotels - Scoring function \(\mathcal{F}_1\)]{
        \includegraphics[width=0.32\textwidth,height=0.18\textheight]{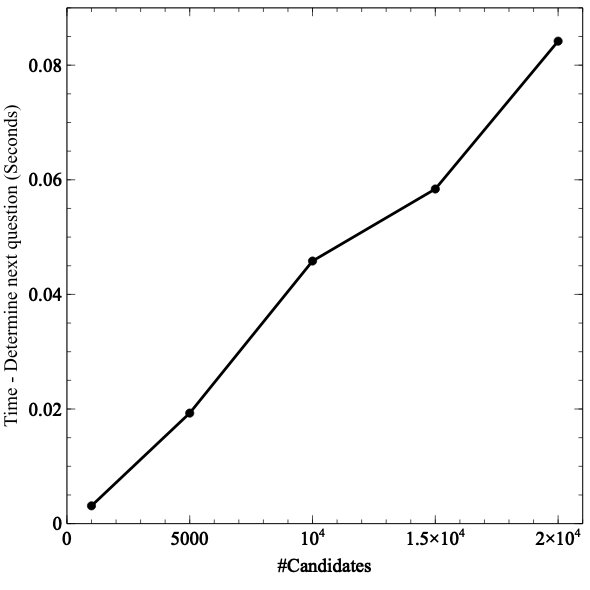}
    }\hfill
    \subfigure[Movies - Scoring function \(\mathcal{F}_3\)]{
        \includegraphics[width=0.32\textwidth,height=0.18\textheight]{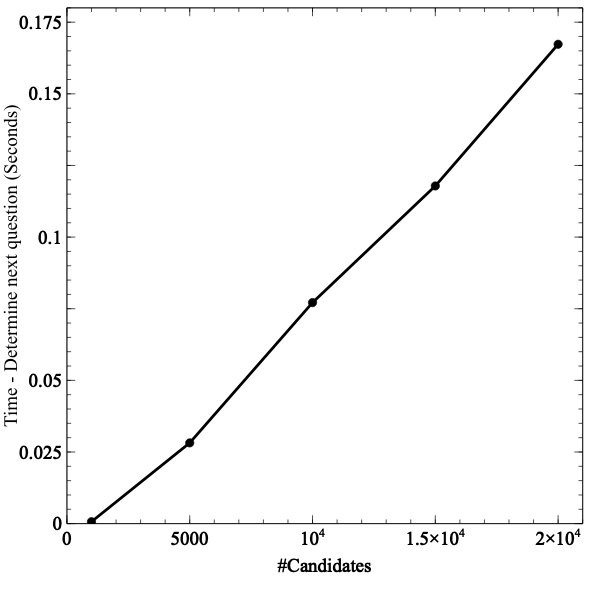}
    }\hfill
    \subfigure[Businesses - Scoring function \(\mathcal{F}_5\)]{
        \includegraphics[width=0.32\textwidth,height=0.18\textheight]{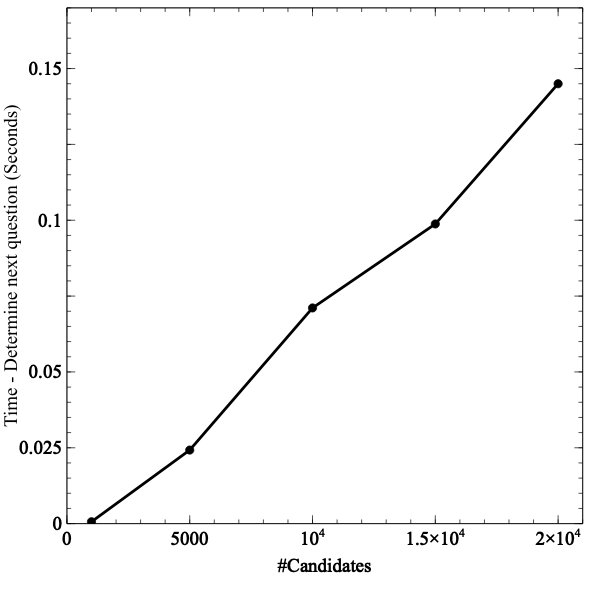}
    }
    \caption{\small Time taken for {\tt EntrRed} using {\tt ProbInd} to determine the next question varying \#candidates.}
    \label{fig:scalibility_determine_next_q}
     \vspace{-0.1in}
\end{figure*}

\begin{figure*}[!htbp]
    \centering
    \subfigure[Hotels - Scoring function \(\mathcal{F}_1\)]{
        \includegraphics[width=0.32\textwidth,height=0.18\textheight]{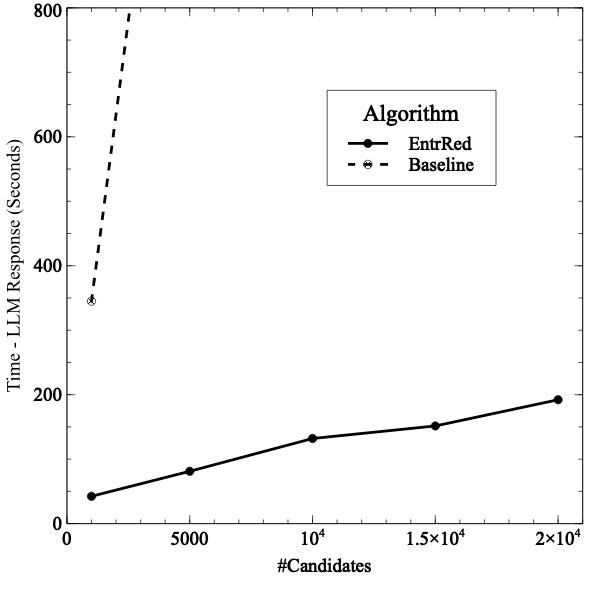}
    }\hfill
    \subfigure[Movies - Scoring function \(\mathcal{F}_4\)]{
        \includegraphics[width=0.32\textwidth,height=0.18\textheight]{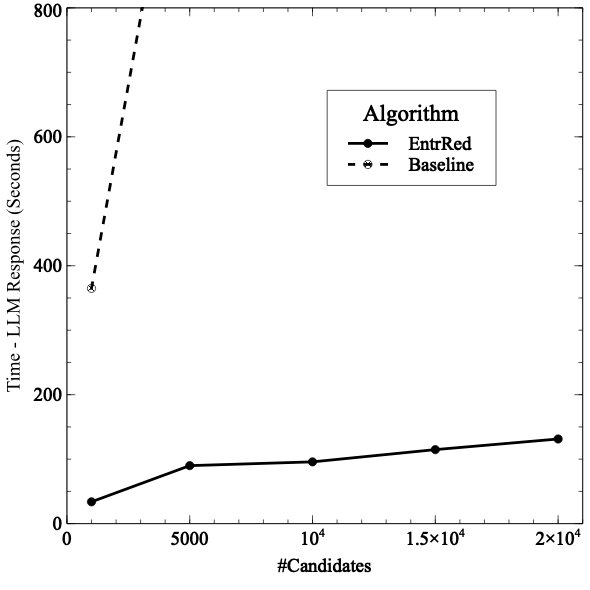}
    }\hfill
    \subfigure[Businesses - Scoring function \(\mathcal{F}_5\)]{
        \includegraphics[width=0.32\textwidth,height=0.18\textheight]{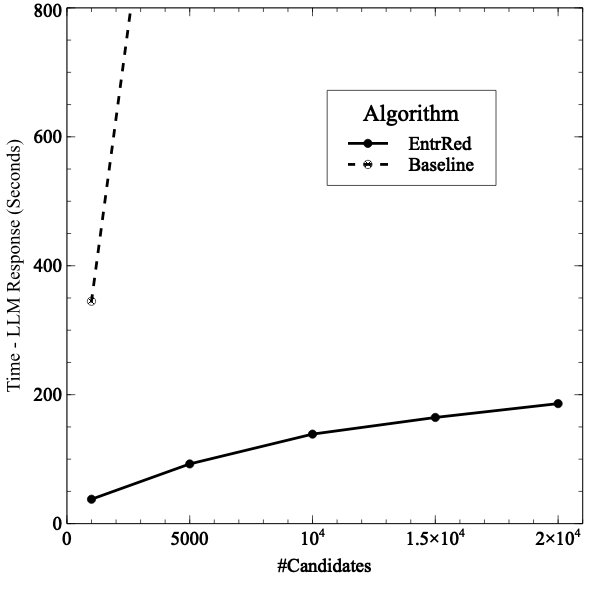}
    }\hfill
    \caption{\small Time taken for LLM response varying \#candidates.}
    \label{fig:scalibility_llm_response}
     \vspace{-0.1in}
\end{figure*}

\section{Related Work}\label{sec:rel}
Our work is closest to reducing the cost of ML inference and to combining queries and ML inference since we aim to answer queries over predicted scores. We believe we are the first to examine the question of top-$k$ under this setting. 
One line of work in query inference, provides native relational support for ML  using containerized solutions such as Amazon Aurora~\cite{sagemaker_2022}, or in-application solutions such as Google's BigQuery ML~\cite{bigquery_2022} and Microsoft's Raven~\cite{DBLP:conf/cidr/KaranasosIPSPPX20}. 

Recent work adopted the use of cheap proxy models, such as image classifiers, to identify  an  approximate  set of $k$ entities satisfying a query~\cite{DBLP:journals/pvldb/DingAL22,DBLP:journals/pvldb/KangGBHZ20}. These works propose sampling approaches to achieve a target precision or recall accuracy. 
Lai et al.~\cite{10.1145/3448016.3452786} studies approximate top-$k$ queries with light-weight proxy models that generate oracle label distribution.
Gao et al.~\cite{DBLP:conf/sigmod/GaoXAY21} introduce a Multi-Level Splitting Sampling to let one "promising" sample path prefix generate multiple "offspring" paths, and direct Monte-Carlo based simulations toward more promising paths.  The ThalamusDB system is an approximate query processing system that integrates into SQL natural language predicates on multimodal data: visual, audio, and text~\cite{DBLP:journals/pacmmod/JoT24}. The system chooses optimized plans according to user preferences on approximation error, computation time, and user labeling overheads. This work is complementary to ours as it does not handle ranked retrieval.
Bolukbasi et al.~\cite{DBLP:conf/icml/BolukbasiWDS17} studied incremental predictions for neural networks.  
 Computation time is reduced by pruning examples that are classified in earlier layers, selected adaptively. 
Kang et al.~\cite{DBLP:journals/pvldb/KangEABZ17} present NOSCOPE, a system for querying videos that can reduce the cost of neural network video analysis by up to three orders of magnitude via inference-optimized model search.
Lu et al.~\cite{DBLP:conf/sigmod/LuCKC18} and Yang et al.~\cite{10.14778/3547305.3547310} use probabilistic predicates to filter data blobs that do not satisfy the query and empirically increase data reduction rates. 
Anderson et al.~\cite{DBLP:conf/icde/AndersonCRW19} use a hierarchical model to reduce the runtime cost of queries over visual content. 
\section{Extensions}\label{sec:ext}
This work opens up several intriguing avenues for future exploration, some of which we discuss below. 

\noindent {\bf Asking Multiple Questions.}
Our approach focused on selecting the next best question. However, an alternative approach would involve determining the next best set of $k'$ questions. This would necessitate adapting the second and third tasks in our framework. We are actively pursuing this extension.

\noindent {\bf Response Processing.}
We assumed a single discrete response from one oracle. However, we recognize that other possibilities exist, such as:
a. An oracle providing a range of values rather than a discrete response,
b. Multiple oracles each offering a discrete response,
c. Multiple oracles each providing a range of responses.

For case (a), this could be treated as a uniform probability distribution, and we could adapt the bound computation task to handle such distributions. The rest of the framework would not require significant changes. Handling (b) and (c) presents additional challenges. As an example, they could be treated as a problem of combining multiple probability distributions, and classical techniques, such as, convolution of probability distributions~\cite{olds1952note} could be used. Alternatively, machine learning methods such as Reinforcement Learning with Human Feedback (RLHF)~\cite{rlhf1, rlhf2, rlhf3, rlhf4, rlhf5} could also be applied. We plan to explore these possibilities in future work.

\noindent {\bf Multiple Noisy Oracles.}
An interesting extension involves the scenario where multiple noisy oracles are present, each with different costs associated with querying. This would require a substantial modification of our framework. Instead of minimizing the number of questions asked, the focus would shift to minimizing the total cost of querying while still maximizing the likelihood of correctly identifying the top-$k$ entities.

\noindent {\bf Rank-based Querying.}
Our framework is designed to identify the top-$k$ entities based on a user-defined scoring function. A natural next step is to extend this to determine the top-$k$ entities in a ranked order. Adapting our framework to handle ranked queries would involve modifications to the first three tasks. We hope to explore this extension as part of our future work as well.

\section{Conclusion}\label{sec:conc}
We presented the novel problem of computing personalized top-$k$ sets with user-defined set-based scoring functions when the function constructs are computed with expensive oracle calls. This problem is ubiquitous in applications where scoring requires access to multimodal entities such as pictures and audio. Computing scores of those entities requires calling oracles such as LLMs. We believe our work to be the first to examine the question of reducing the computation cost of score constructs for set-based scoring functions when invoking expert oracles.

Our framework aims to reduce the number of oracle calls necessary to compute the exact top-$k$ set. It relies on four tasks that compute score bounds of candidate sets and a probabilistic model that takes into account interdependency between those sets. The framework determines the next best question to ask an LLM to complete scores and handles score uncertainty obtained from LLMs. We make principled contributions in designing solution for each of these tasks.

Our extensive experiments on a large corpora of entities representing hotels and their modalities, demonstrate that our framework is able to compute exact top-$k$ results by invoking the minimum number of calls necessary to expensive oracles. This work opens many new directions that we discussed in Section~\ref{sec:ext}.

%\clearpage

\bibliographystyle{ACM-Reference-Format}
\bibliography{sample}


\begin{thebibliography}{28}

%%% the Web address.

\ifx \showCODEN    \undefined \def \showCODEN     #1{\unskip}     \fi
\ifx \showDOI      \undefined \def \showDOI       #1{#1}\fi
\ifx \showISBNx    \undefined \def \showISBNx     #1{\unskip}     \fi
\ifx \showISBNxiii \undefined \def \showISBNxiii  #1{\unskip}     \fi
\ifx \showISSN     \undefined \def \showISSN      #1{\unskip}     \fi
\ifx \showLCCN     \undefined \def \showLCCN      #1{\unskip}     \fi
\ifx \shownote     \undefined \def \shownote      #1{#1}          \fi
\ifx \showarticletitle \undefined \def \showarticletitle #1{#1}   \fi
\ifx \showURL      \undefined \def \showURL       {\relax}        \fi
% invisible to TeX
\providecommand\bibfield[2]{#2}
\providecommand\bibinfo[2]{#2}
\providecommand\natexlab[1]{#1}
\providecommand\showeprint[2][]{arXiv:#2}

\bibitem[\protect\citeauthoryear{Anderson, Cafarella, Ros, and Wenisch}{Anderson et~al\mbox{.}}{2019}]%
        {DBLP:conf/icde/AndersonCRW19}
\bibfield{author}{\bibinfo{person}{Michael~R. Anderson}, \bibinfo{person}{Michael~J. Cafarella}, \bibinfo{person}{German Ros}, {and} \bibinfo{person}{Thomas~F. Wenisch}.} \bibinfo{year}{2019}\natexlab{}.
\newblock \showarticletitle{Physical Representation-Based Predicate Optimization for a Visual Analytics Database}. In \bibinfo{booktitle}{\emph{35th {IEEE} International Conference on Data Engineering, {ICDE} 2019, Macao, China, April 8-11, 2019}}. \bibinfo{pages}{1466--1477}.
\newblock


\bibitem[\protect\citeauthoryear{Authors}{Authors}{[n.d.]}]%
        {tr}
\bibfield{author}{\bibinfo{person}{Authors}.} \bibinfo{year}{[n.d.]}\natexlab{}.
\newblock \bibinfo{title}{Tech Report}.
\newblock \bibinfo{howpublished}{\url{https://www.dropbox.com/scl/fo/z4lkx4e3iw7npe7ahvfum/AIwQF08WzJzglLg7pmBXM6A?rlkey=dkvb7sl8rjy8zk1w2cx976f04&dl=0}}.
\newblock


\bibitem[\protect\citeauthoryear{Bolukbasi, Wang, Dekel, and Saligrama}{Bolukbasi et~al\mbox{.}}{2017}]%
        {DBLP:conf/icml/BolukbasiWDS17}
\bibfield{author}{\bibinfo{person}{Tolga Bolukbasi}, \bibinfo{person}{Joseph Wang}, \bibinfo{person}{Ofer Dekel}, {and} \bibinfo{person}{Venkatesh Saligrama}.} \bibinfo{year}{2017}\natexlab{}.
\newblock \showarticletitle{Adaptive Neural Networks for Efficient Inference}. In \bibinfo{booktitle}{\emph{Proceedings of the 34th International Conference on Machine Learning, {ICML} 2017, Sydney, NSW, Australia, 6-11 August 2017}} \emph{(\bibinfo{series}{Proceedings of Machine Learning Research})}, \bibfield{editor}{\bibinfo{person}{Doina Precup} {and} \bibinfo{person}{Yee~Whye Teh}} (Eds.), Vol.~\bibinfo{volume}{70}. \bibinfo{publisher}{{PMLR}}, \bibinfo{pages}{527--536}.
\newblock


\bibitem[\protect\citeauthoryear{Borodin, Jain, Lee, and Ye}{Borodin et~al\mbox{.}}{2017}]%
        {div3}
\bibfield{author}{\bibinfo{person}{Allan Borodin}, \bibinfo{person}{Aadhar Jain}, \bibinfo{person}{Hyun~Chul Lee}, {and} \bibinfo{person}{Yuli Ye}.} \bibinfo{year}{2017}\natexlab{}.
\newblock \showarticletitle{Max-sum diversification, monotone submodular functions, and dynamic updates}.
\newblock \bibinfo{journal}{\emph{ACM Transactions on Algorithms (TALG)}} \bibinfo{volume}{13}, \bibinfo{number}{3} (\bibinfo{year}{2017}), \bibinfo{pages}{1--25}.
\newblock


\bibitem[\protect\citeauthoryear{Das}{Das}{2024}]%
        {arnab_das_2024}
\bibfield{author}{\bibinfo{person}{Arnab Das}.} \bibinfo{year}{2024}\natexlab{}.
\newblock \bibinfo{title}{Hotels Dataset}.
\newblock
\newblock
\urldef\tempurl%
\url{https://doi.org/10.34740/KAGGLE/DSV/7542286}
\showDOI{\tempurl}


\bibitem[\protect\citeauthoryear{Das, Ivkin, Bansal, Rouesnel, Gautier, Karnin, Dirac, Ramakrishnan, Perunicic, Shcherbatyi, Wu, Zolic, Shen, Ahmed, Winkelmolen, Miladinovic, Archembeau, Tang, Dutt, Grao, and Venkateswar}{Das et~al\mbox{.}}{2020}]%
        {sagemaker_2022}
\bibfield{author}{\bibinfo{person}{Piali Das}, \bibinfo{person}{Nikita Ivkin}, \bibinfo{person}{Tanya Bansal}, \bibinfo{person}{Laurence Rouesnel}, \bibinfo{person}{Philip Gautier}, \bibinfo{person}{Zohar Karnin}, \bibinfo{person}{Leo Dirac}, \bibinfo{person}{Lakshmi Ramakrishnan}, \bibinfo{person}{Andre Perunicic}, \bibinfo{person}{Iaroslav Shcherbatyi}, \bibinfo{person}{Wilton Wu}, \bibinfo{person}{Aida Zolic}, \bibinfo{person}{Huibin Shen}, \bibinfo{person}{Amr Ahmed}, \bibinfo{person}{Fela Winkelmolen}, \bibinfo{person}{Miroslav Miladinovic}, \bibinfo{person}{Cedric Archembeau}, \bibinfo{person}{Alex Tang}, \bibinfo{person}{Bhaskar Dutt}, \bibinfo{person}{Patricia Grao}, {and} \bibinfo{person}{Kumar Venkateswar}.} \bibinfo{year}{2020}\natexlab{}.
\newblock \showarticletitle{Amazon SageMaker Autopilot: A White Box AutoML Solution at Scale}. In \bibinfo{booktitle}{\emph{Proceedings of the Fourth International Workshop on Data Management for End-to-End Machine Learning}} (Portland, OR, USA) \emph{(\bibinfo{series}{DEEM'20})}. \bibinfo{publisher}{Association for Computing Machinery}, \bibinfo{address}{New York, NY, USA}, Article \bibinfo{articleno}{2}, \bibinfo{numpages}{7}~pages.
\newblock
\showISBNx{9781450380232}
\urldef\tempurl%
\url{https://doi.org/10.1145/3399579.3399870}
\showDOI{\tempurl}


\bibitem[\protect\citeauthoryear{Ding, Amer{-}Yahia, and Lakshmanan}{Ding et~al\mbox{.}}{2022}]%
        {DBLP:journals/pvldb/DingAL22}
\bibfield{author}{\bibinfo{person}{Dujian Ding}, \bibinfo{person}{Sihem Amer{-}Yahia}, {and} \bibinfo{person}{Laks V.~S. Lakshmanan}.} \bibinfo{year}{2022}\natexlab{}.
\newblock \showarticletitle{On Efficient Approximate Queries over Machine Learning Models}.
\newblock \bibinfo{journal}{\emph{Proc. {VLDB} Endow.}} \bibinfo{volume}{16}, \bibinfo{number}{4} (\bibinfo{year}{2022}), \bibinfo{pages}{918--931}.
\newblock
\urldef\tempurl%
\url{https://doi.org/10.14778/3574245.3574273}
\showDOI{\tempurl}


\bibitem[\protect\citeauthoryear{Fernandes and Bernardino}{Fernandes and Bernardino}{2015}]%
        {bigquery_2022}
\bibfield{author}{\bibinfo{person}{S\'{e}rgio Fernandes} {and} \bibinfo{person}{Jorge Bernardino}.} \bibinfo{year}{2015}\natexlab{}.
\newblock \showarticletitle{What is BigQuery?}. In \bibinfo{booktitle}{\emph{Proceedings of the 19th International Database Engineering \&amp; Applications Symposium}} (Yokohama, Japan) \emph{(\bibinfo{series}{IDEAS '15})}. \bibinfo{publisher}{Association for Computing Machinery}, \bibinfo{address}{New York, NY, USA}, \bibinfo{pages}{202–203}.
\newblock
\showISBNx{9781450334143}
\urldef\tempurl%
\url{https://doi.org/10.1145/2790755.2790797}
\showDOI{\tempurl}


\bibitem[\protect\citeauthoryear{Gao, Xu, Agarwal, and Yang}{Gao et~al\mbox{.}}{2021}]%
        {DBLP:conf/sigmod/GaoXAY21}
\bibfield{author}{\bibinfo{person}{Junyang Gao}, \bibinfo{person}{Yifan Xu}, \bibinfo{person}{Pankaj~K. Agarwal}, {and} \bibinfo{person}{Jun Yang}.} \bibinfo{year}{2021}\natexlab{}.
\newblock \showarticletitle{Efficiently Answering Durability Prediction Queries}. In \bibinfo{booktitle}{\emph{{SIGMOD} '21: International Conference on Management of Data, Virtual Event, China, June 20-25, 2021}}. \bibinfo{pages}{591--604}.
\newblock


\bibitem[\protect\citeauthoryear{Gollapudi and Sharma}{Gollapudi and Sharma}{2009}]%
        {div4}
\bibfield{author}{\bibinfo{person}{Sreenivas Gollapudi} {and} \bibinfo{person}{Aneesh Sharma}.} \bibinfo{year}{2009}\natexlab{}.
\newblock \showarticletitle{An axiomatic approach for result diversification}. In \bibinfo{booktitle}{\emph{Proceedings of the 18th international conference on World wide web}}. \bibinfo{pages}{381--390}.
\newblock


\bibitem[\protect\citeauthoryear{Hong, Lee, and Thorne}{Hong et~al\mbox{.}}{2024}]%
        {rlhf3}
\bibfield{author}{\bibinfo{person}{Jiwoo Hong}, \bibinfo{person}{Noah Lee}, {and} \bibinfo{person}{James Thorne}.} \bibinfo{year}{2024}\natexlab{}.
\newblock \showarticletitle{Orpo: Monolithic preference optimization without reference model}. In \bibinfo{booktitle}{\emph{Proceedings of the 2024 Conference on Empirical Methods in Natural Language Processing}}. \bibinfo{pages}{11170--11189}.
\newblock


\bibitem[\protect\citeauthoryear{Jo and Trummer}{Jo and Trummer}{2024}]%
        {DBLP:journals/pacmmod/JoT24}
\bibfield{author}{\bibinfo{person}{Saehan Jo} {and} \bibinfo{person}{Immanuel Trummer}.} \bibinfo{year}{2024}\natexlab{}.
\newblock \showarticletitle{ThalamusDB: Approximate Query Processing on Multi-Modal Data}.
\newblock \bibinfo{journal}{\emph{Proc. {ACM} Manag. Data}} \bibinfo{volume}{2}, \bibinfo{number}{3} (\bibinfo{year}{2024}), \bibinfo{pages}{186}.
\newblock
\urldef\tempurl%
\url{https://doi.org/10.1145/3654989}
\showDOI{\tempurl}


\bibitem[\protect\citeauthoryear{Kaminskas and Bridge}{Kaminskas and Bridge}{2016}]%
        {div2}
\bibfield{author}{\bibinfo{person}{Marius Kaminskas} {and} \bibinfo{person}{Derek Bridge}.} \bibinfo{year}{2016}\natexlab{}.
\newblock \showarticletitle{Diversity, serendipity, novelty, and coverage: a survey and empirical analysis of beyond-accuracy objectives in recommender systems}.
\newblock \bibinfo{journal}{\emph{ACM Transactions on Interactive Intelligent Systems (TiiS)}} \bibinfo{volume}{7}, \bibinfo{number}{1} (\bibinfo{year}{2016}), \bibinfo{pages}{1--42}.
\newblock


\bibitem[\protect\citeauthoryear{Kang, Emmons, Abuzaid, Bailis, and Zaharia}{Kang et~al\mbox{.}}{2017}]%
        {DBLP:journals/pvldb/KangEABZ17}
\bibfield{author}{\bibinfo{person}{Daniel Kang}, \bibinfo{person}{John Emmons}, \bibinfo{person}{Firas Abuzaid}, \bibinfo{person}{Peter Bailis}, {and} \bibinfo{person}{Matei Zaharia}.} \bibinfo{year}{2017}\natexlab{}.
\newblock \showarticletitle{NoScope: Optimizing Deep CNN-Based Queries over Video Streams at Scale}.
\newblock \bibinfo{journal}{\emph{Proc. {VLDB} Endow.}} \bibinfo{volume}{10}, \bibinfo{number}{11} (\bibinfo{year}{2017}), \bibinfo{pages}{1586--1597}.
\newblock


\bibitem[\protect\citeauthoryear{Kang, Gan, Bailis, Hashimoto, and Zaharia}{Kang et~al\mbox{.}}{2020}]%
        {DBLP:journals/pvldb/KangGBHZ20}
\bibfield{author}{\bibinfo{person}{Daniel Kang}, \bibinfo{person}{Edward Gan}, \bibinfo{person}{Peter Bailis}, \bibinfo{person}{Tatsunori Hashimoto}, {and} \bibinfo{person}{Matei Zaharia}.} \bibinfo{year}{2020}\natexlab{}.
\newblock \showarticletitle{Approximate Selection with Guarantees using Proxies}.
\newblock \bibinfo{journal}{\emph{Proc. {VLDB} Endow.}} \bibinfo{volume}{13}, \bibinfo{number}{11} (\bibinfo{year}{2020}), \bibinfo{pages}{1990--2003}.
\newblock


\bibitem[\protect\citeauthoryear{Karanasos, Interlandi, Psallidas, Sen, Park, Popivanov, Xin, Nakandala, Krishnan, Weimer, Yu, Ramakrishnan, and Curino}{Karanasos et~al\mbox{.}}{2020}]%
        {DBLP:conf/cidr/KaranasosIPSPPX20}
\bibfield{author}{\bibinfo{person}{Konstantinos Karanasos}, \bibinfo{person}{Matteo Interlandi}, \bibinfo{person}{Fotis Psallidas}, \bibinfo{person}{Rathijit Sen}, \bibinfo{person}{Kwanghyun Park}, \bibinfo{person}{Ivan Popivanov}, \bibinfo{person}{Doris Xin}, \bibinfo{person}{Supun Nakandala}, \bibinfo{person}{Subru Krishnan}, \bibinfo{person}{Markus Weimer}, \bibinfo{person}{Yuan Yu}, \bibinfo{person}{Raghu Ramakrishnan}, {and} \bibinfo{person}{Carlo Curino}.} \bibinfo{year}{2020}\natexlab{}.
\newblock \showarticletitle{Extending Relational Query Processing with {ML} Inference}. In \bibinfo{booktitle}{\emph{10th Conference on Innovative Data Systems Research, {CIDR} 2020, Amsterdam, The Netherlands, January 12-15, 2020, Online Proceedings}}. \bibinfo{publisher}{www.cidrdb.org}.
\newblock


\bibitem[\protect\citeauthoryear{Kotkov, Wang, and Veijalainen}{Kotkov et~al\mbox{.}}{2016}]%
        {div1}
\bibfield{author}{\bibinfo{person}{Denis Kotkov}, \bibinfo{person}{Shuaiqiang Wang}, {and} \bibinfo{person}{Jari Veijalainen}.} \bibinfo{year}{2016}\natexlab{}.
\newblock \showarticletitle{A survey of serendipity in recommender systems}.
\newblock \bibinfo{journal}{\emph{Knowledge-Based Systems}}  \bibinfo{volume}{111} (\bibinfo{year}{2016}), \bibinfo{pages}{180--192}.
\newblock


\bibitem[\protect\citeauthoryear{Lai, Han, Liu, Zhang, Lo, and Kao}{Lai et~al\mbox{.}}{2021}]%
        {10.1145/3448016.3452786}
\bibfield{author}{\bibinfo{person}{Ziliang Lai}, \bibinfo{person}{Chenxia Han}, \bibinfo{person}{Chris Liu}, \bibinfo{person}{Pengfei Zhang}, \bibinfo{person}{Eric Lo}, {and} \bibinfo{person}{Ben Kao}.} \bibinfo{year}{2021}\natexlab{}.
\newblock \showarticletitle{Top-K Deep Video Analytics: A Probabilistic Approach}. In \bibinfo{booktitle}{\emph{Proceedings of the 2021 International Conference on Management of Data}} (Virtual Event, China) \emph{(\bibinfo{series}{SIGMOD '21})}. \bibinfo{publisher}{Association for Computing Machinery}, \bibinfo{address}{New York, NY, USA}, \bibinfo{pages}{1037–1050}.
\newblock
\showISBNx{9781450383431}
\urldef\tempurl%
\url{https://doi.org/10.1145/3448016.3452786}
\showDOI{\tempurl}


\bibitem[\protect\citeauthoryear{Lee, Phatale, Mansoor, Mesnard, Ferret, Lu, Bishop, Hall, Carbune, Rastogi, et~al\mbox{.}}{Lee et~al\mbox{.}}{[n.d.]}]%
        {rlhf5}
\bibfield{author}{\bibinfo{person}{Harrison Lee}, \bibinfo{person}{Samrat Phatale}, \bibinfo{person}{Hassan Mansoor}, \bibinfo{person}{Thomas Mesnard}, \bibinfo{person}{Johan Ferret}, \bibinfo{person}{Kellie~Ren Lu}, \bibinfo{person}{Colton Bishop}, \bibinfo{person}{Ethan Hall}, \bibinfo{person}{Victor Carbune}, \bibinfo{person}{Abhinav Rastogi}, {et~al\mbox{.}}} \bibinfo{year}{[n.d.]}\natexlab{}.
\newblock \showarticletitle{RLAIF vs. RLHF: Scaling Reinforcement Learning from Human Feedback with AI Feedback}. In \bibinfo{booktitle}{\emph{Forty-first International Conference on Machine Learning}}.
\newblock


\bibitem[\protect\citeauthoryear{Lu, Chowdhery, Kandula, and Chaudhuri}{Lu et~al\mbox{.}}{2018}]%
        {DBLP:conf/sigmod/LuCKC18}
\bibfield{author}{\bibinfo{person}{Yao Lu}, \bibinfo{person}{Aakanksha Chowdhery}, \bibinfo{person}{Srikanth Kandula}, {and} \bibinfo{person}{Surajit Chaudhuri}.} \bibinfo{year}{2018}\natexlab{}.
\newblock \showarticletitle{Accelerating Machine Learning Inference with Probabilistic Predicates}. In \bibinfo{booktitle}{\emph{Proceedings of the 2018 International Conference on Management of Data, {SIGMOD} Conference 2018, Houston, TX, USA, June 10-15, 2018}}. \bibinfo{pages}{1493--1508}.
\newblock


\bibitem[\protect\citeauthoryear{Melnyk, Mroueh, Belgodere, Rigotti, Nitsure, Yurochkin, Greenewald, Navratil, and Ross}{Melnyk et~al\mbox{.}}{2024}]%
        {rlhf4}
\bibfield{author}{\bibinfo{person}{Igor Melnyk}, \bibinfo{person}{Youssef Mroueh}, \bibinfo{person}{Brian Belgodere}, \bibinfo{person}{Mattia Rigotti}, \bibinfo{person}{Apoorva Nitsure}, \bibinfo{person}{Mikhail Yurochkin}, \bibinfo{person}{Kristjan Greenewald}, \bibinfo{person}{Jiri Navratil}, {and} \bibinfo{person}{Jerret Ross}.} \bibinfo{year}{2024}\natexlab{}.
\newblock \showarticletitle{Distributional Preference Alignment of LLMs via Optimal Transport}.
\newblock \bibinfo{journal}{\emph{arXiv preprint arXiv:2406.05882}} (\bibinfo{year}{2024}).
\newblock


\bibitem[\protect\citeauthoryear{Olds}{Olds}{1952}]%
        {olds1952note}
\bibfield{author}{\bibinfo{person}{Edwin~G Olds}.} \bibinfo{year}{1952}\natexlab{}.
\newblock \showarticletitle{A note on the convolution of uniform distributions}.
\newblock \bibinfo{journal}{\emph{The Annals of Mathematical Statistics}} \bibinfo{volume}{23}, \bibinfo{number}{2} (\bibinfo{year}{1952}), \bibinfo{pages}{282--285}.
\newblock


\bibitem[\protect\citeauthoryear{Rafailov, Sharma, Mitchell, Manning, Ermon, and Finn}{Rafailov et~al\mbox{.}}{2024}]%
        {rlhf2}
\bibfield{author}{\bibinfo{person}{Rafael Rafailov}, \bibinfo{person}{Archit Sharma}, \bibinfo{person}{Eric Mitchell}, \bibinfo{person}{Christopher~D Manning}, \bibinfo{person}{Stefano Ermon}, {and} \bibinfo{person}{Chelsea Finn}.} \bibinfo{year}{2024}\natexlab{}.
\newblock \showarticletitle{Direct preference optimization: Your language model is secretly a reward model}.
\newblock \bibinfo{journal}{\emph{Advances in Neural Information Processing Systems}}  \bibinfo{volume}{36} (\bibinfo{year}{2024}).
\newblock


\bibitem[\protect\citeauthoryear{Rahman, Thirumuruganathan, Roy, Amer-Yahia, and Das}{Rahman et~al\mbox{.}}{2015}]%
        {rahman2015worker}
\bibfield{author}{\bibinfo{person}{Habibur Rahman}, \bibinfo{person}{Saravanan Thirumuruganathan}, \bibinfo{person}{Senjuti~Basu Roy}, \bibinfo{person}{Sihem Amer-Yahia}, {and} \bibinfo{person}{Gautam Das}.} \bibinfo{year}{2015}\natexlab{}.
\newblock \showarticletitle{Worker skill estimation in team-based tasks}.
\newblock \bibinfo{journal}{\emph{Proceedings of the VLDB Endowment}} \bibinfo{volume}{8}, \bibinfo{number}{11} (\bibinfo{year}{2015}), \bibinfo{pages}{1142--1153}.
\newblock


\bibitem[\protect\citeauthoryear{R{\'e}nyi}{R{\'e}nyi}{1961}]%
        {renyi1961measures}
\bibfield{author}{\bibinfo{person}{Alfr{\'e}d R{\'e}nyi}.} \bibinfo{year}{1961}\natexlab{}.
\newblock \showarticletitle{On measures of entropy and information}. In \bibinfo{booktitle}{\emph{Proceedings of the fourth Berkeley symposium on mathematical statistics and probability, volume 1: contributions to the theory of statistics}}, Vol.~\bibinfo{volume}{4}. University of California Press, \bibinfo{pages}{547--562}.
\newblock


\bibitem[\protect\citeauthoryear{Schulman, Wolski, Dhariwal, Radford, and Klimov}{Schulman et~al\mbox{.}}{2017}]%
        {rlhf1}
\bibfield{author}{\bibinfo{person}{John Schulman}, \bibinfo{person}{Filip Wolski}, \bibinfo{person}{Prafulla Dhariwal}, \bibinfo{person}{Alec Radford}, {and} \bibinfo{person}{Oleg Klimov}.} \bibinfo{year}{2017}\natexlab{}.
\newblock \showarticletitle{Proximal policy optimization algorithms}.
\newblock \bibinfo{journal}{\emph{arXiv preprint arXiv:1707.06347}} (\bibinfo{year}{2017}).
\newblock


\bibitem[\protect\citeauthoryear{Yang, Wang, Huang, Lu, Li, and Wang}{Yang et~al\mbox{.}}{2022}]%
        {10.14778/3547305.3547310}
\bibfield{author}{\bibinfo{person}{Zhihui Yang}, \bibinfo{person}{Zuozhi Wang}, \bibinfo{person}{Yicong Huang}, \bibinfo{person}{Yao Lu}, \bibinfo{person}{Chen Li}, {and} \bibinfo{person}{X.~Sean Wang}.} \bibinfo{year}{2022}\natexlab{}.
\newblock \showarticletitle{Optimizing Machine Learning Inference Queries with Correlative Proxy Models}.
\newblock \bibinfo{journal}{\emph{Proc. VLDB Endow.}} \bibinfo{volume}{15}, \bibinfo{number}{10} (\bibinfo{date}{jun} \bibinfo{year}{2022}), \bibinfo{pages}{2032–2044}.
\newblock
\showISSN{2150-8097}
\urldef\tempurl%
\url{https://doi.org/10.14778/3547305.3547310}
\showDOI{\tempurl}


\bibitem[\protect\citeauthoryear{Zhou, Chai, Li, and SUN}{Zhou et~al\mbox{.}}{2020}]%
        {9094012}
\bibfield{author}{\bibinfo{person}{Xuanhe Zhou}, \bibinfo{person}{Chengliang Chai}, \bibinfo{person}{Guoliang Li}, {and} \bibinfo{person}{JI SUN}.} \bibinfo{year}{2020}\natexlab{}.
\newblock \showarticletitle{Database Meets Artificial Intelligence: A Survey}.
\newblock \bibinfo{journal}{\emph{IEEE Transactions on Knowledge and Data Engineering}} \bibinfo{volume}{1}, \bibinfo{number}{1} (\bibinfo{year}{2020}), \bibinfo{pages}{1--18}.
\newblock


\end{thebibliography}

\end{document}